%% file: main.tex
\documentclass[letterpaper, 11pt]{article}
\usepackage[margin=1in]{geometry}
\usepackage{latexsym}
\usepackage{amssymb}
\usepackage{times}
\usepackage{xcolor}
\usepackage{color}
\usepackage{amsmath,amsfonts,amsthm}
\usepackage{graphicx,float}
\usepackage{enumitem}
\usepackage[ruled,vlined,resetcount, algosection]{algorithm2e}
\usepackage{setspace}

\usepackage{hyperref}
\makeatletter
\def\namedlabel#1#2{\begingroup
    #2%
    \def\@currentlabel{#2}%
    \phantomsection\label{#1}\endgroup
}
\makeatother

\hypersetup{
bookmarks=true,
colorlinks=true,
linkcolor=blue,
urlcolor=blue,
citecolor=blue,
pdftex,
linktocpage=true, % makes the page number as hyperlink in table of content
hyperindex=true
}

\newcommand{\poly}{\mathsf{poly}}

\newcommand{\eps}{\varepsilon}

\newtheorem{lemma}{Lemma}[section]
\newtheorem{theorem}[lemma]{Theorem}

\newtheorem{construction}[lemma]{Construction}

\newtheorem{definition}[lemma]{Definition}
\newtheorem{corollary}[lemma]{Corollary}

\newtheorem{remark}[lemma]{Remark}

\newtheorem{algo}[lemma]{Algorithm}

\begin{document}

\allowdisplaybreaks

\begin{titlepage}
\def\thepage{}

\title{Block Edit Errors with Transpositions: Deterministic Document Exchange Protocols and Almost Optimal Binary Codes}
\author{Kuan Cheng  \thanks{kcheng17@jhu.edu.\ Department of Computer Science, Johns Hopkins University. Supported by NSF Grant CCF-1617713.} \and Zhengzhong Jin  \thanks{zjin12@jhu.edu.\ Department of Computer Science, Johns Hopkins University.\ Partially supported by NSF Grant CCF-1617713.}\and Xin Li  \thanks{lixints@cs.jhu.edu.\ Department of Computer Science, Johns Hopkins University. Supported by NSF Grant CCF-1617713.} \and Ke Wu \thanks{AshleyMo@jhu.edu.\ Department of Computer Science, Johns Hopkins University.} }

\maketitle \thispagestyle{empty}

%%%%%%%%%%%%%%%%%%%%%%%%%%%%%%%%%%

\input{abstract.tex}

\end{titlepage}

\input{intro.tex}

\input{prelim.tex}

\input{BTprotocol.tex}
\input{uniformrandomprotocol.tex}

\input{BTcode.tex}

\section{Acknowledgements}
We thank an anonymous referee for catching a bug in the previous version of this paper, and Bernhard Haeupler for very useful feedbacks.

\bibliographystyle{plain}
\bibliography{ref}

\newpage
\input{appendix.tex}

\end{document}

%% file: abstract.tex
Document exchange and error correcting codes are two fundamental problems regarding communications. In the first problem, Alice and Bob each holds a string, and the goal is for Alice to send a short sketch to Bob, so that Bob can recover Alice's string. In the second problem, Alice sends a message with some redundant information to Bob through a channel that can add adversarial errors, and the goal is for Bob to correctly recover the message despite the errors. In both problems, an upper bound is placed on the number of errors between the two strings or that the channel can add, and a major goal is to minimize the size of the sketch or the redundant information. In this paper we focus on deterministic document exchange protocols and binary error correcting codes.

Both problems have been studied extensively.\ In the case of Hamming errors (i.e., bit substitutions) and bit erasures, we have explicit constructions with asymptotically optimal parameters. However, other error types are still rather poorly understood.\ In a recent work \cite{CJLW18}, the authors constructed explicit deterministic document exchange protocols and binary error correcting codes for edit errors with almost optimal parameters.\ Unfortunately, the constructions in \cite{CJLW18} do not work for other common errors such as block transpositions. 

In this paper, we generalize the constructions in \cite{CJLW18} to handle a much larger class of errors. These include bursts of insertions and deletions, as well as block transpositions. Specifically, we consider document exchange and error correcting codes where the total number of block insertions, block deletions, and block transpositions is at most $k \leq \alpha n/\log n$ for some constant $0<\alpha<1$. In addition, the total number of bits inserted and deleted by the first two kinds of operations is at most $t \leq \beta n$ for some constant $0<\beta<1$, where $n$ is the length of Alice's string or message. We construct explicit, deterministic document exchange protocols with sketch size $ O( (k  \log n +t) \log^2 \frac{n}{k\log n + t} )$ and explicit binary error correcting code with $O(k \log n \log \log \log n+t)$ redundant bits. As a comparison, the information-theoretic optimum for both problems is $\Theta(k \log n+t)$. As far as we know, previously there are no known explicit deterministic document exchange protocols in this case, and the best known binary code needs $\Omega(n)$ redundant bits even to correct just \emph{one} block transposition \cite{796406}.\footnote{We note that by combining the techniques in \cite{HS17c} and \cite{haeupler2017synsimucode}, one can get an explicit binary code that corrects $k$ block transpositions with $\widetilde{O}(\sqrt{kn})$ redundant bits. However to our knowledge this result has not appeared anywhere in the literature, and moreover it requires at least $\widetilde{\Omega}(\sqrt{n})$ redundant bits even to correct one block transposition.}%However, previous explicit constructions are quite far from achieving this. In particular, there are no known deterministic document exchange protocols in this case. The only previous protocol is a randomized protocol by Irmak et al. \cite{} which has sketch size $O(k \log^2 n)$. For binary error correcting codes, the only previous work that can handle edit errors and block transpositions is the work by Schulman and Zuckerman \cite{}, which can handle $\Omega(n)$ insertions and deletions, and $\Omega(n/\log n)$ block transpositions with $O(n)$ redundant bits. No construction is known for smaller $k$, which means one needs $O(n)$ redundant bits even to handle just \emph{one} block transposition.\ In this work we significantly improve the situation.\ Specifically, we construct an explicit, deterministic document exchange protocol with sketch size $O(k \log^2 n)$, which matches the randomized protocol of \cite{} and is optimal up to an additional $\log n$ factor. We then use this to construct an explicit error correcting code with redundancy $\tilde{O}(k \log n)$, which is optimal up to an additional $\log \log n$ factor.

%% file: intro.tex
\section{Introduction}
In communications and more generally distributed computing environments, questions arises regarding the synchronization of files or messages. For example, a message sent from one party to another party through a channel may get modified by channel noise or adversarial errors, and files stored on distributed servers may become out of sync due to different edit operations by different users. In many situations, these questions can be formalized in the framework of the following two fundamental problems.

\begin{itemize}
\item \emph{Document exchange.} In this problem, two parties Alice and Bob each holds a string $x$ and $y$, and the two strings are within distance $k$ in some metric space. The goal is for Alice to send a short sketch to Bob, so that Bob can recover $x$ based on his string $y$ and the sketch. %Naturally, we would like to require both the message length and the computation time of Alice and Bob to be as small as possible.

\item \emph{Error correcting codes.} In this problem, two parties Alice and Bob are linked by a channel, which can change any string sent into another string within distance $k$ in some metric space. Alice's goal is to send a message to Bob. She does this by sending an encoding of the message through the channel, which contains some redundant information, so that Bob can recover the correct message despite any changes to the codeword. %Again, we would like to minimize both the codeword length (or equivalently, the number of redundant bits) and the encoding/decoding time. This is a generalization of the classical error correcting codes for Hamming errors.
\end{itemize}

These two problems are closely related. For example, in many cases a solution to the document exchange problem can also be used to construct an error correcting code, but the reverse direction is not necessarily true. In both problems, a major goal is to is to minimize the size of the sketch or the redundant information. For applications in computer science, we also require the computations of both parties to be efficient, i.e., in polynomial time of the input length. In this case we say that the solutions to these problems are \emph{explicit}. Here we focus on deterministic document exchange protocols and error correcting codes with a binary alphabet, arguably the most important setting in computer science.

Both problems have been studied extensively, but the known solutions and our knowledge vary significantly depending on the distance metric in these problems. In the case of Hamming distance (or Hamming errors),  we have a near complete understanding and explicit constructions with asymptotically optimal parameters. However, for other distance metrics/error types, our understanding is still rather limited.

An important generalization of Hamming errors is edit errors, which consist of bit insertions and deletions. These are strictly more general than Hamming errors since a bit substitution can be replaced by a deletion followed by an insertion. Edit errors can happen in many practical situations, such as reading magnetic and optical media, mutations in gene sequences, and routing packets in Internet protocols. However, these errors are considerably harder to handle, due to the fact that a single edit error can change the positions of all the bits in a string.

Non-explicitly, by using a greedy graph coloring algorithm or a sphere packing argument, one can show that the optimal size of the sketch in document exchange, or the redundant information in error correcting codes is roughly the same for both Hamming errors and edit errors. Specifically, suppose that Alice's string or message has length $n$ and the distance bound $k$ is relatively small (e.g., $k \leq n/4$), then for both Hamming errors and edit errors, the optimal size in both problems is $\Theta(k \log (\frac{n}{k}))$ \cite{Levenshtein66}. For Hamming errors, this can be achieved by using sophisticated linear Algebraic Geometric codes \cite{hoholdt1998algebraic}, but for edit errors the situation is quite different. We now describe some of the previous works regarding both document exchange and error correcting codes for edit errors.

\paragraph{Document exchange.} Orlitsky \cite{185373} first studied the document exchange problem for generally correlated strings $x, y$. Using the greedy graph coloring algorithm mentioned before, he obtained a deterministic protocol with sketch size $O(k \log n)$ for edit errors, but the running time is exponential in $k$. Subsequent improvements appeared in \cite{CormodePSV00}, \cite{irmak2005improved},  and \cite{Jowhari2012EfficientCP}, achieving sketch size $O(k \log(\frac{n}{k}) \log n)$ \cite{irmak2005improved} and $O(k \log^2 n \log^* n)$ \cite{Jowhari2012EfficientCP} with running time $\tilde{O}(n)$. A recent work by Chakraborty et al. \cite{Chakraborty2015LowDE} further obtained sketch size $O(k^2 \log n)$ and running time $\tilde{O}(n)$, by using a clever randomized embedding from the edit distance metric to the Hamming distance metric. Based on this work, Belazzougui and Zhang \cite{BelazzouguiZ16} gave an improved protocol with sketch size $O(k (\log^2 k+\log n))$,  which is asymptotically optimal for $k =2^{O(\sqrt{\log n})}$. The running time in \cite{BelazzouguiZ16} is $\tilde{O}(n+\poly(k))$.

Unfortunately, all of the above protocols, except the one in \cite{185373} which runs in exponential time, are randomized. Although randomized protocols are still useful in practice, having deterministic ones would certainly bring much more benefits. Furthermore, randomized protocols are also not suitable for the applications in constructing error correcting codes. However, designing an efficient deterministic protocol appears quite tricky, and it was not until 2015 when Belazzougui \cite{Belazzougui2015EfficientDS} gave the first deterministic protocol even for $k>1$. The protocol in \cite{Belazzougui2015EfficientDS} has sketch size $O(k^2 + k \log^2 n)$ and running time  $\tilde{O}(n)$.

\paragraph{Error correcting codes.} As fundamental objects in both theory and practice, error correcting codes have been studied extensively from the pioneering work of Shannon and Hamming. While great success has been achieved in constructing codes for Hamming errors, the progress on codes for edit errors has been quite slow despite much research. A work by Levenshtein \cite{Levenshtein66} in 1966 showed that the Varshamov-Tenengolts code \cite{VT65}  corrects one deletion with an optimal redundancy of roughly $\log n$ bits, but even correcting two deletions requires $\Omega(n)$ redundant bits. In 1999, Schulman and Zuckerman \cite{796406} gave an explicit asymptotically good code, that can correct up to $\Omega(n)$ edit errors with $O(n)$ redundant bits. However the same amount of redundancy is needed even for smaller number of errors. For more earlier works on this subject, we refer the reader to the survey by Mercier et al. \cite{MercierBT10}.%and even the first explicit asymptotically good insdel code (a code that has constant rate and can also correct a constant fraction of insertion and deletion errors) over a constant alphabet did not appear until the work of  Schulman and Zuckerman in 1999 \cite{796406}, who gave such a code over the binary alphabet. We refer the reader to the survey by Mercier et al. \cite{MercierBT10} for more details about the extensive research on this topic.

In recent years there have been several works trying to improve the situation. Specifically, a line of work by Guruswami et.\ al \cite{7835185}, \cite{7541373}, \cite{BukhV16} constructed explicit codes that can correct $1-\eps$ fraction of edit errors with rate $\Omega(\eps^5)$ and alphabet size $\poly(1/\eps)$; and codes that can correct $1-\frac{2}{t+1}-\eps$ fraction of errors with rate $(\eps/t)^{\poly(1/\eps)}$ for a fixed alphabet size $t \geq 2$. Another line of work by Haeupler et al. \cite{haeupler2017synchronization}, \cite{HS17c}, \cite{CHLSW18} introduced and constructed a combinatorial object called \emph{synchronization string}, which can be used to transform standard error correcting codes into codes for edit errors by increasing the alphabet size.\ Via this transformation, \cite{haeupler2017synchronization} achieved explicit codes that can correct $\delta$ fraction of edit errors with rate $1-\delta-\eps$ and alphabet size exponential in $\frac{1}{\eps}$, which approaches the singleton bound. All of these works however require a relatively large alphabet size. %although the alphabet size is exponential in $\frac{1}{\eps}$.

In the case of binary alphabets, for any fixed constant $k$, a recent work by Brakensiek et.\ al \cite{8022906} constructed an explicit code that can correct $k$ edit errors with $O(k^2 \log k \log n)$ redundant bits. This is asymptotically optimal when $k$ is a fixed constant, but the construction in \cite{8022906} only works for constant $k$, and breaks down for larger $k$ (e.g., $k=\log n$). Based on his deterministic document exchange protocol, Belazzougui \cite{Belazzougui2015EfficientDS} also gave an explicit code that can correct up to $k$ edit errors with $O(k^2 + k \log^2 n)$ redundant bits. Finally, the work by Haeupler et. al \cite{haeupler2017synsimucode} constructed explicit codes that can correct $\delta$ fraction of edit errors with rate $1-\Theta(\sqrt{\delta \log(1/\delta)})$, whereas the (non-explicit) optimal rate is $1-\Theta(\delta \log(1/\delta))$.

In a very recent work by the authors \cite{CJLW18}, we significantly improved the situation. Specifically, we constructed an explicit document exchange protocol with sketch size $O(k \log^2 \frac{n}{k}) $, which is optimal except for an additional $\log \frac{n}{k}$ factor. This also implies an explicit binary code that can correct $\delta$ fraction of edit errors with rate $1-\Theta(\delta \log^2(1/\delta))$, which is optimal up to an additional $\log(1/\delta)$ factor. These two results are also independently obtained by Haeupler \cite{haeupler2018optimal}. We also constructed explicit codes for $k$ edit errors with $O(k \log n)$ redundant bits, which is optimal for $k \leq n^{1-\alpha}$, any constant $0<\alpha<1$. These results bring our understanding of document exchange and error correcting codes for edit errors much closer to that of standard Hamming errors.

However, the constructions in \cite{CJLW18} and \cite{haeupler2018optimal} do not work for other common types of errors, such as \emph{block transpositions}. Given any string $x$, a block transposition takes an arbitrary substring $z$ of $x$, cuts it to make $x$ become $\tilde{x}$, and then finds a different position in $\tilde{x}$ and insert $z$ as a block into $\tilde{x}$. These errors happen frequently in distributed file systems and Internet protocols. For example, it is quite common that a user, when editing a file, moves a whole paragraph in the file to somewhere else; and in Internet routing protocols, packets can often get rearranged during the process. Block transpositions also arise naturally in biological processes, where a subsequence of genes can be moved in one step during mutation. In the setting of document exchange or error correcting codes, it is easy to see that even a single transposition of a block with length $t$ can result in $2t$ edit errors, thus a naive application of document exchange protocols or codes for edit errors will result in very bad parameters.

In this paper we consider document exchange protocols and error correcting codes for edit errors as well as block transpositions. In fact, even for edit errors we also consider a larger, more general class of errors. Specifically, we consider edit errors that happen in \emph{bursts}. This kind of errors is also pretty common, as most errors that happen in practice, such as in wireless or mobile communications and magnetic disk readings, tend to be concentrated. We model such errors as \emph{block} insertions and deletions, where in one operation the adversary can insert or delete a whole block of bits. It is again easy to see that this is indeed a generalization of standard edit errors. Therefore, in this paper we study three \emph{block edit} operations: block insertions, block deletions, and block transpositions. However, in addition to the bound $k$ on such operations, we also need to put a bound on the total number of bits that the adversary can insert or delete, since otherwise the adversary can simply delete the whole string in one block deletion. Therefore, we model the adversary as follows.

\paragraph{Model of the adversary.} For some parameters $k$ and $t$ and an alphabet $\Sigma$, a $(k, t)$ block edit adversary is allowed to perform three kinds of operations: block insertion, block deletion and block transposition. The adversary is allowed to perform at most $k$ such operations, while the total number of symbols inserted/deleted by the first two operations is at most $t$. We also use $(k, t)$ block edit errors to denote errors introduced by such an adversary. All our results focus on the case of binary alphabet, but in our protocols and analysis we will be using larger alphabets.

We note that by the result of Schulman and Zuckerman \cite{796406}, to correct $\Omega(n/\log n)$ block transpositions one needs at least $\Omega(n)$ redundant bits. Thus we only consider $k \leq \alpha n/\log n$ for some constant $0< \alpha<1$. Similarly, we only consider $t \leq \beta n$ for  some constant $0< \beta<1$ since otherwise the adversary can simply delete the whole string. We also note the following subtle difference between the three block edit operations. While we need a bound $t$ on the total number of bits that the adversary can insert or delete, for block transposition an adversary can choose to move an \emph{arbitrarily long} substring. Therefore, we need to consider the three operations separately, and cannot simply replace a block transposition by a block deletion followed by a block insertion.

Edit errors with block transpositions have been studied before in several different contexts. For example, Shapira and Storer \cite{ShapiraS02} showed that finding the distance between two given strings under this metric is $\mathsf{NP}$-hard, and they gave an efficient algorithm that achieves $O(\log n)$ approximation. Interestingly, a work by Cormode and Muthukrishnan \cite{CormodeM07} showed that this metric can be embedded into the $\mathsf{L_1}$ metric with distortion $O(\log n \log^* n)$; and they used it to give a near linear time algorithm that achieves $O(\log n \log^* n)$ approximation for this distance, something currently unknown for the standard edit distance. Coming back to document exchange and error correcting codes, in our model, we show in the appendix that non-explicitly, the information optimum for both the sketch size of document exchange, and the redundancy of error correcting codes, is $\Theta(k \log n+t)$.

\paragraph{Related previous work on block transpositions.} When it comes to more general errors such as block transpositions, as far as we know, there are no known explicit deterministic document exchange protocols. The only known randomized protocols which can handle edit errors as well as block transpositions are the protocol of \cite{irmak2005improved}, which has sketch size $O(k \log(\frac{n}{k}) \log n)$; and the protocol of \cite{Jowhari2012EfficientCP}, which has sketch size $\tilde{O}(k \log^2 n)$. The protocol of \cite{irmak2005improved} uses a recursive tree structure and random hash functions, while the protocol of \cite{Jowhari2012EfficientCP} is based on the embedding of Cormode and Muthukrishnan \cite{CormodeM07}. We stress that both of these protocols are randomized, and there are very good reasons why it is not easy to modify them into deterministic ones. Specifically, unlike in our previous work \cite{CJLW18} and the work of Haeupler \cite{haeupler2018optimal}, a direct derandomization of the hash functions used in \cite{irmak2005improved} (for example by using almost $k$-wise independent sample space) does \emph{not} give a deterministic protocol, because block transpositions will make the computation of a matching problematic. We shall discuss this in more details when we give an overview of our techniques. On the other hand, the embedding of Cormode and Muthukrishnan \cite{CormodeM07} results in an exponentially large dimension, thus directly sending a sketch deterministically will result in a prohibitively large size. This is why the protocol of \cite{Jowhari2012EfficientCP} has to perform a dimension reduction first, which is necessarily randomized.

Similarly, the only previous explicit codes that can handle edit errors as well as block transpositions are the work of Schulman and Zuckerman \cite{796406}, and the work of Haeupler et al.  \cite{HS17c}. Both can recover from $\Omega(n/\log n)$ block transpositions with $\Omega(n)$ redundant bits (\cite{HS17c} can also recover from block replications), but \cite{796406} has a binary alphabet while \cite{HS17c} has a constant size alphabet. However the work of Schulman and Zuckerman \cite{796406} also needs $\Omega(n)$ redundant bits even to correct one block transposition. We further note that by combining the techniques in \cite{HS17c} and \cite{haeupler2017synsimucode}, one can get an explicit binary code that corrects $k$ block transpositions with $\widetilde{O}(\sqrt{kn})$ redundant bits. However to our knowledge this result has not appeared anywhere in the literature, and moreover it requires at least $\widetilde{\Omega}(\sqrt{n})$ redundant bits even to correct one block transposition. We note that however none of the previous works mentioned studied edit errors that can allow block insertions/deletions.

\subsection{Our results}
In this paper we construct explicit  document exchange protocols, and error correcting codes for adversaries discussed above. We have the following theorems.
\begin{theorem}
There exist constants $\alpha, \beta \in (0, 1)$ such that for every $n, k, t \in \mathbb{N}$ with $k \leq \alpha n/\log n, t \leq \beta n$, there exists an explicit  binary document exchange protocol with sketch size   $ O( (k  \log n +t) \log^2 \frac{n}{k\log n + t} )$, against a $(k, t)$ block edit adversary. %adversary who can perform at most $k$ block insertions/deletions and block transpositions, where the total number of bits inserted/deleted is at most $t$. %=k_1 + k_2$, for $(k_1, t)$ block-insertions/deletions and $k_2$ block-transpositions.
\end{theorem}

This is the first explicit binary document exchange protocol for block edit errors.\ The sketch size matches the randomized protocols of \cite{irmak2005improved} and \cite{Jowhari2012EfficientCP} up to an additional $\log \frac{n}{k\log n + t }$ factor, and is optimal up to an additional $\log^2 \frac{n}{k\log n + t }$ factor. Using this protocol, we can construct the following error correcting code.

\begin{theorem}
There exist constants $\alpha, \beta \in (0, 1)$ such that for every $n, k, t \in \mathbb{N}$ with $k \leq \alpha n/\log n, t \leq \beta n$, there exists an explicit binary error correcting code with message length $n$ and codeword length $n+  O( (k  \log n +t) \log^2 \frac{n}{k\log n + t} )$, against a $(k, t)$ block edit adversary. %that can correct up to $k$ block insertions/deletions and block transpositions, where the total number of bits inserted/deleted is at most $t$.
\end{theorem}

For small $k, t$ we can actually achieve the following result, which gives better parameters. %Next we have the following theorem of error correcting codes.

\begin{theorem}
There exist constants $\alpha, \beta \in (0, 1)$ such that for every $n, k, t \in \mathbb{N}$ with $k \leq \alpha n/\log n, t \leq \beta n$, there exists an explicit binary code with message length $n$ and codeword length $n+O(k\log n\log\log \log n+t)$, against a $(k, t)$ block edit adversary. %that can correct up to $k$ block insertions/deletions and block transpositions, where the total number of bits inserted/deleted is at most $t$.
\end{theorem}

In the case of small $k, t$, these results significantly improve the result of Schulman and Zuckerman \cite{796406}, which needs $\Omega(n)$ redundant bits even to correct one block transposition, and the result obtained by combining the techniques in \cite{HS17c} and \cite{haeupler2017synsimucode}, which needs $\widetilde{\Omega}(\sqrt{n})$ redundant bits even to correct one block transposition. The redundancy here is also optimal up to an extra $\log \log \log n$ factor or $\log^2 \frac{n}{k\log n + t }$ factor.

As a special case, we obtain the following corollaries for standard edit errors with block transpositions.

\begin{corollary}
There exist a constant $\alpha \in (0, 1)$ such that for every $n, k \in \mathbb{N}$ with $k \leq \alpha n/\log n$, there exists an explicit binary document exchange protocol with sketch size $O(k \log n \log^2\frac{n}{k \log n }) $, against an adversary who can perform $k$ edit operations or block transpositions.
\end{corollary}

\begin{corollary}
There exist a constant $\alpha \in (0, 1)$ such that for every $n, k \in \mathbb{N}$ with $k \leq \alpha n/\log n$, there exists an explicit binary error correcting code with message length $n$ and codeword length $\mathsf{min}\{n+O(k \log n \log^2\frac{n}{k \log n }), n+O(k\log n\log\log \log n)\}$, against an adversary who can perform $k$ edit operations or block transpositions.
\end{corollary}

\begin{remark}
As illustrated by our theorems and corollaries, the sketch size in our document exchange protocol or the number of redundant bits in our error correcting codes do \emph{not} depend on the size of a block in block transpositions, they only depend on the number of such operations performed.\ In contrast, the sketch size or the number of redundant bits do depend on the size of a block in block insertions or deletions. This again shows that we cannot simply treat a block transposition as a block deletion followed by a block insertion, because that will lead to a sketch size dependent on the block size.
\end{remark}

\subsection{Overview of our techniques}

In this section we provide an informal, high-level overview of our techniques. %Our construction is based on the basic recursive tree structure used in \cite{irmak2005improved} and improved in \cite{CJLW18}, and a more sophisticated way to \emph{approximate} maximum matchings in the computation.
%the framework of our previous work \cite{CJLW18}.
One important difference between this work and previous works is that in this work, we cannot use several recently introduced synchronization techniques, such as synchronization strings \cite{haeupler2017synchronization}, self-matching hash functions \cite{CJLW18}, or synchronization hash functions \cite{CJLW18}. The reason is that synchronization strings are designed for relatively large alphabets (e.g., constant size), and often result in worse parameters when translating into the binary alphabet; while self-matching hash functions and synchronization hash functions are specifically tailored for standard edit errors, and they break down once block transpositions are allowed. Instead, for document exchange we rely on the basic recursive tree structure used in \cite{irmak2005improved} and improved in \cite{CJLW18}, together with a new and more sophisticated way to \emph{approximate} maximum \emph{non-monotone, non-overlapping} matchings in the computation; and for error correcting codes we combine the string parsing techniques in \cite{CormodeM07} with our framework in \cite{CJLW18}. We start giving more details by describing our document exchange protocol.

\paragraph{Document exchange.} We first briefly describe the construction in \cite{CJLW18}. The protocol has $O(\log \frac{n}{k})$ levels where $k$ is the number of edit errors between Alice's string $x$ and Bob's string $y$.  Throughout the protocol, Bob always maintains a string $\tilde{x}$ (his current guess of Alice's string $x$). In the $i$-th level, both Alice and Bob partition their strings $x$ and $\tilde{x}$ evenly into $O(2^{i} k)$ blocks, i.e., in each subsequent level they divide a block in the previous level evenly into two blocks.\ The following invariance is maintained: in each level, at most $ck$ blocks are different between $x$ and $\tilde{x}$, where $c$ is a universal constant.

This property is satisfied at the beginning of the protocol, and maintained for subsequent levels as follows: in each level Alice constructs an appropriate hash function based on her string $x$. This function has a short description. Alice then hashes every block of $x$ and sends some redundancy of the hash values together with the description of the hash function to Bob. The redundancy here is computed by a systematic error correcting code that can correct $ck$ Hamming errors, whose alphabet corresponds to the output of the hash function. Bob, after receiving the redundancy of the hash values and the short description, first uses the hash function to hash every block of $\tilde{x}$. Since $\tilde{x}$ and $x$ differ in at most  $c k$ blocks, Bob can use the redundancy to correctly recover all the hash values. He then uses dynamic programming to find a maximum \emph{monotone} matching between $x$ and $y$ under the hash function and the hash values, and uses the matched blocks of $y$ to fill the corresponding blocks of his string $\tilde{x}$. The analysis shows that the Bob can correctly recover all blocks of $x$ except at most $c k/2$ of them. Thus in the next level $\tilde{x}$ is the same as $x$ except for at most $c k$ blocks. At the end of the protocol when the size of each block has become small enough (i.e., $O(\log \frac{n}{k})$), Alice can just send a sketch for $ck$ Hamming errors of the blocks to let Bob finally recover $x$.

Our starting point is to try to generalize the above protocol. However, one immediate difficulty is to handle block transpositions. The protocol of  \cite{CJLW18} actually performs badly for such errors. To see this consider the following example: the adversary simply moves the first $0.4 n$ bits of $x$ to the end.  Since the protocol in \cite{CJLW18} tries to find the maximum monotone matching in each level, Bob can only recover the last $0.6 n$ bits of $x$ since this gives the maximum monotone matching. In this case, one single error has cost roughly half of the string; while as a comparison, for standard edit errors, the protocol in \cite{CJLW18} lets Bob recover all except $O(1)$ blocks if there is only one edit error.

To resolve this issue, we  make several important changes to the protocol in \cite{CJLW18}. The first major change is that, in each level, instead of having Bob find the maximum \emph{monotone} matching between $x$ and $y$ using the hash values, we let Bob find the maximum \emph{non-monotone} matching. However, the hash functions used in \cite{CJLW18} are not suitable for this purpose, since the hash functions there actually allow a small number of collisions in the hash values of blocks of $x$, and the use of these hash functions in \cite{CJLW18} relies crucially on the property of a monotone matching. Instead, here we strengthen the hash function to ensure that there is no collision, by using a slightly larger output size. We call such hash functions \emph{collision free} hash functions.
\begin{definition}[Collision free hash functions]

Given $n, p, q \in  \mathbb{N}, p \leq n$ and a string $x \in \{0,1\}^n$, we say a function $h:\{0,1\}^{p} \rightarrow \{0,1\}^{q}$ is collision free (for $x$), if for every $i,j\in [n - p+1]$, $h(x[i, i+p)) = h(x[j, j+p))$ if and only if $ x[i, i+p) = x[j, j+p) $. Here $x[i, j)$ denotes the substring of $x$ which starts at the $i$'th bit and ends at the $j-1$'th bit.

\end{definition}
This definition guarantees that if the hash function we used is collision free, then any two different substrings of $x$ cannot have the same hash values.

We show that a collision free hash function can be constructed by using a $\frac{1}{\poly n}$-almost $ O(\log n)$-wise independent generator with seed length (number of random bits used) $O(\log n)$. This can work since for each pair of distinct substrings, their hash values are the same with probability $1/\poly(n)$. Since there are at most $O(n^2)$ pairs, a union bound shows the existence of collision free hash functions. To get a deterministic hash function, we check each possible seed to see if the corresponding hash function is collision free, which can be done by checking if every pair of different substrings of $x$ have different hash values. Note that there are at most $O(n^2)$ pairs and the seed length of the generator is $O(\log n)$, so this can be done in polynomial time. %This can work since for each pair of distinct substrings, their hash values are the same with probability $1/\poly(n)$. And there are at most $O(n^2)$ pairs. So a union bound can show such a good hash function exists.

%Our protocol inherits the levelled structure of \cite{CJLW18}, having $O(\log \frac{n}{k\log n + t})$ levels.
%In the $i$-th level, Alice computes a good hash function $h_i$ using $x$. Then she divides $x$ evenly into $O(2^i (k+\frac{t}{\log n}))$ blocks and hashes every block. She sends the redundancy of the hash values to Bob together with the short description of $h_i$.
%For Bob's side, for the $i$-th level, he receives $h_i$ and the
However, even a non-monotone matching under collision free hash functions is not enough for our purpose. The reason is that in the matching, we are trying to match every well divided block of $x$ to every possible block of $y$ (not necessarily the blocks obtained by dividing $y$ evenly into disjoint blocks), because we have edit errors here. If we just do this in the naive way, then the matched blocks of $y$ can be \emph{overlapping}. Using these overlapping blocks of $y$ to fill the blocks of $\tilde{x}$ is problematic, since even a single edit error or block transposition can create many new (overlapping) blocks in $y$ (which can be as large as the length of the block in each level). These new blocks are all possible to be matched, and then we won't be able to maintain an upper bound of $O(k)$ on the different blocks between $x$ and $\tilde{x}$.

To solve this, we need to insist on computing a maximum \emph{non-overlapping}, non-monotone matching.
\begin{definition}[Non-overlapping (non-monotone) matching]

Given $n, n', p, q \in \mathbb{N}, p\leq n, p\leq n'$, a function $h:\{0,1\}^{p} \rightarrow \{0,1\}^{q}$ and two strings $x\in \{0,1\}^{n}, y\in \{0,1\}^{n'}$, a (non-overlapping) matching between $x$ and $y$ under $h$ is a sequence of matches (pairs of indices) $w = ((i_1, j_1), \ldots (i_{|w|}, j_{|w|}))$ s.t.
\begin{itemize}
\item for every $k\in [|w|]$,
\begin{itemize}
\item $i_k = 1+ p l_k \in [n]$ for some $l_k$, i.e., each $i_k$ is the starting index of some block of $x$, when $x$ is divided evenly into disjoint blocks of length $p$,
\item $j_{k} \in [n']$,
\item $ h(x[i_k, {i_k} + p)) = h(y[j_k ,{j_k} + p))$.
\end{itemize}

\item $i_1, \ldots, i_{|w|}$ are distinct.

\item Intervals $ [j_k, j_k+ p ), k\in [|w|]  $, are disjoint.

\end{itemize}
\end{definition}

Under this definition, we can indeed show a similar upper bound on the number of different blocks between $x$ and $\tilde{x}$ in each level, if Bob finds the maximum non-overlapping matching. However, another technical difficulty arises: how to compute a maximum non-overlapping, non-monotone matching efficiently. This is unclear since the standard algorithm to compute a maximum matching only gives a possibly overlapping matching, while the dynamic programming approach in \cite{CJLW18} only works for a monotone matching. %whether there is a polynomial time algorithm to find this matching. We give a positive answer to this question by a using a max-flow model and a semi-definite programming (SDP).

Computing the maximum non-overlapping, non-monotone matching turns out to be a hard task, and we were not able to find an efficient algorithm that accomplishes this exactly.\ Instead, we consider an algorithm that \emph{approximates} the maximum non-overlapping, non-monotone matching. However, this raises several other issues. The first issue is how to maintain the invariance that in each level $x$ and $\tilde{x}$ only differ in a small number of blocks. For example, consider level $i$ and assuming $x$ is partitioned into $l_i$ blocks, then we would like Bob to obtain a matching of size at least $l_i - O(k + t/\log n)$ (recall $t$ is the total number of bits inserted or deleted). Thus if $k$ and $t$ are small then even a $0.99$ approximation is still far from achieving our goal.

To get around this, we modify the protocol so that in each level Bob only computes a matching for the blocks that are unmatched in the previous level or detected to be incorrectly matched in this level (the detection can be done by comparing the hash values of the block and its matched block). If the number of such blocks can be bounded by some $O(k + t/\log n )$, then we only need a constant factor approximation. To keep the invariance in each level, note that the approximation factor should be larger than $1/2$ since each unmatched or incorrectly matched block will become two blocks in the next level.

Unfortunately, directly achieving such an approximation still seems hard. Thus we further relax the problem to allow some slight overlaps in the matching, i.e., we require that each bit of Bob's string $y$ appears in at most $d$ matched pairs in each level for some small number $d$ (e.g., a constant or $\log n$). We call this a \emph{degree d overlapping} matching (note that a non-overlapping matching is simply a degree 1 overlapping matching). Although this may cause extra errors in the matching, we show that the number of incorrectly matched pairs can be bounded by $O((k  + t/\log n) i )$ (instead of $O(k + t/\log n )$) in level $i$.

To achieve this, we first give a $1/3$ approximation algorithm for the maximum non-monotone, non-overlapping matching. Then we give another algorithm that achieves matching size at least $2/3$ of the maximum non-monotone, non-overlapping matching, while this matching obtained is a degree $3$ overlapping matching. For simplicity we also refer to this as a $2/3$ approximation algorithm.

The $1/3$ approximation is obtained by a greedy algorithm, which starts with an empty matching $w$ and visits $x$'s blocks one by one and tries to match it with a substring in $y$ (according to the hash function and hash values), such that the substring does not overlap with any substring in $y$ that is already matched. If such a matched pair is found then it is added to $w$. The algorithm keeps running until it cannot add any more matched pair. %For block $i$, it sweeps $y$, finding an interval that is not overlapped with any intervals in matches of $w$. If the interval can further match the picked block, i.e. they have the same hash value for a given function $h$, then it adds the match $(i, j)$ to $w$ where $j$ is the left end position the interval.

To see this indeed gives a $1/3$ approximation, assume the maximum non-monotone, non-overlapping matching is $w^*$.
Each time the algorithm adds a matched pair to $w$, at most $3$ matched pairs in $w^*$ will be excluded from being added to $w$ since they either have overlaps with $y$'s substring in the added pair or correspond to the same block of $x$. As a result, when $|w| < 1/3|w^*| $, there always exist some matched pairs in $w^*$ that can be added to $w$. Thus, at the end of the algorithm, $|w|\geq 1/3|w^*|$. %So it can find one of these matches and add it to $w$. Thus finally $|w|$ will be at least $|w^*|$.

Next we show a $2/3$ approximation algorithm that gives a degree 3 overlapping matching. The idea is to run the greedy algorithm for $3$ times, where each time the algorithm is applied to unmatched blocks of $x$ and the entire string $y$. To see the approximation factor, again let $w^*$ be the optimal non-monotone, non-overlapping matching.
After the first time, the matching $w$ has size at least $1/3 |w^*| $. So $|w^*|$ will have at least $|w^*| - |w|$ matched pairses for unmatched blocks in $x$. Therefore after the second time, the size of the matching is at least $ |w| + 1/3 (|w^*| - |w|) \geq  5/9 |w^*|$. Similarly, after the third time, the matching will have size at least $2/3 |w^*|$. As the greedy algorithm is applied three times, each bit of $y$ can appear in at most $3$ matched pairses in $w$.

We now bound the number of incorrectly matched and unmatched blocks in each level. First we claim that each non-monotone non-overlapping matching has at most $O(k + t/\log n)$ incorrectly matched blocks.

This is because by our definition of collision free hash function, if a pair is incorrectly matched then the substring of $y$ must contain some edit operation applied to $x$, since otherwise the pair will definitely have different hash values if they are different. Thus we only need to count how many non-overlapping new substrings in $y$ (i.e. those not equal to any substring of $x$) one can get after $(k, t)$ block edit errors. One insertion or deletion of $t_1$ bits will create at most $ O(1) + O(t/\log n) $ new substrings since the block size is always at least $\log n$. One block transposition will create at most $O(1)$ non-overlapping substrings in $y$ that are not equal to any substring of $x$. So in total there are at most $O(k + t/\log n)$ new non-overlapping substrings in $y$. Similarly, it is easy to generalize this claim, and show that each degree $d$ overlapping matching has at most $O(d(k + t/\log n))$ incorrectly matched blocks.

Now to bound the number of incorrectly matched blocks in level i, notice that the matching we obtained in this level is a degree $3i$ overlapping matching, since in each level we compute a degree $3$ overlapping matching using the entire string $y$ and we combine them together. Thus there are at most $O((k+t/\log n) i)$ incorrectly matched blocks.
%As a result, each matching in a level before $i$ will have at most $3 \times O(k + t/\log n) = O(k+t/\log n)$ incorrectly matched pairs at level $i$, even if no wrong pairs are detected in previous levels. So overall there are at most $O((k+t/\log n) i)$ wrong matches.

%The number of unmatched blocks can be bounded by counting how many $y$'s intervals can be deleted or modified by block edit errors. Note that one block insertion/deletion of $t_1$ bits can delete or modify at most $O(1) + t_1/b_i$  blocks of $x$. One block transposition can delete or modify at most
%$O(1)$ blocks of $x$. Remaining blocks can all be matched back to themselves. Thus the number of unmatched blocks is at most $O(k+t/b_i)$.

% if hash values of all previous $i$ levels are correctly recovered by Bob
The number of unmatched blocks can also be upper bounded by $O((k+t/\log n)i)$ using induction. For the base case, the number of blocks in the first level of Bob is at most $l_1 = O(k+t/\log n)$ so the claim holds. Now assume in level $i-1$, the number of unmatched blocks is $c_1(i-1)(k+t/\log n)$, and the number of incorrectly matched blocks is at most $ c_2(i-1)(k+t/\log n)$, for some constants $c_1, c_2$. In level $i$, once Bob recovers all the correct hash values, he can detect some of the incorrectly matched blocks. Let the total number of detected blocks and unmatched blocks be $s$ with $s \leq   (c_1+c_2)(i-1)(k+t/\log n)$. In our algorithm, these blocks are to be rematched in level $i$, and following our previous argument at least $s - c_3(k + t/\log n)$ of them can be matched in the maximum non-overlapping matching for some constant $c_3$.
By our $2/3$ approximation algorithm, the actual matching $w_i$ we get has size at least $ 2/3(s-  c_3(k + t/\log n))$.  Hence the number of unmatched blocks after this is at most $ s - |w_i| \leq 1/3s + 2/3c_3(k+t/\log n) $. We can set $c_1$ to be large enough s.t. this number is still upper bounded by $c_1i(k+t/\log n)$.

As we have bounded the number of incorrectly matched blocks and unmatched blocks by $O(i(k+t/\log n))$ in level $i$, at the beginning of level $i+1$, Alice can send the redundancy of the hash values of her blocks using a code that corrects $O(i(k+t/\log n))$ errors. This allows Bob to recover all the hash values correctly, and the size of the redundancy is $O(i(k+t/\log n)\log n)$ since the hash function outputs $O(\log n)$ bits. We start the protocol with a block size of $O(\frac{n}{k\log n + t})$ and thus the protocol takes $L=O(\log \frac{n}{k\log n + t})$ levels.\ A straightforward computation gives that the sketch size of our protocol is $O( (k\log n +t) \log^2  \frac{n}{k\log n + t})$.
%In the last level where the block size has become $O(\log n)$, Alice can just send a redundancy of size $O(L(k+t/\log n)\log n)$ and Bob can recover all blocks of $x$ correctly.
%The sketch size of our protocol is $O( (k\log n +t) \log^2  \frac{n}{k\log n + t})$. This is due to a straightforward computation. The sketch size for the code of level $i$ is $O((k\log n+t)i)$. Summing up all levels, it is $O((k\log n +t)L^2)$. The sketch size for the final encoding is  $O((k\log n+t)L)$. The length of the description for each hash function is $O(\log n)$. Note that $L =  \frac{n}{k\log n + t}$. Thus the overall sketch length is as stated.

\input{overview_code.tex}

\subsection{Discussions and open problems.} In this paper we study document exchange protocols and error correcting codes for block edit errors. We give the first explicit, deterministic document exchange protocol in this case, and significantly improved error correcting codes. In particular, for both document exchange and error correcting codes, our sketch size or redundant information is close to optimal.

The obvious open problem is to try to achieve truly optimal constructions, where an interesting intermediate step is to try to adapt the $\eps$-self matching hash functions and the $\eps$-synchronization hash functions in \cite{CJLW18} to  handle block transpositions. More broadly, it would be interesting to study document exchange protocols and error correcting codes for other more general errors.

\paragraph{Organization of the paper.} The rest of the paper is organized as follows. In Section \ref{BTprotocol} we give a deterministic document exchange protocol for block insertions/deletions and block transpositions. In Section \ref{randBTProtocol} we give a document exchange protocol for block insertions/deletions and block transpositions for uniformly random strings. Then in Section \ref{sec:BTcode} we give constructions of codes correcting block insertions/deletions and block transpositions. Finally we give tight bounds of the sketch size or redundancy in Appendix \ref{appendix}.

%% file: overview_code.tex
\paragraph{Error Correcting Codes.}
We now describe how to construct an error correcting code 
from a document exchange protocol for block edit errors. Similar to the construction in \cite{CJLW18}, our starting point is to first encode the sketch of the document exchange protocol using the asymptotically good code by Schulman and Zuckerman \cite{796406}, which can resist edit errors and block transpositions.
Then we concatenate the message with the encoding of the sketch.
When decoding, we first decode the sketch, then apply the document exchange protocol on Bob's side to recover the message using the sketch.
%The construction of the Error Correcting Code in \cite{CJLW18} uses the same idea.

However, here we have an additional issue with this approach:  a block transposition may move some part of the encoding of the sketch to somewhere in the middle of the message, or vice versa. In this case, we won't be able to tell which part of the received string is the encoding of the sketch, and which part of the received string is the original message.
%And when we decoding, we don't know which part comes from the message and which part comes from the encoding of the sketch.

To solve this issue, we use a fixed string $\mathsf{buf} = 0^{\ell_\mathsf{buf}} \circ 1$ as a buffer to mark the encoding of the sketch, for some $\ell_\mathsf{buf}=O(\log n)$. More specifically, we evenly divide the encoding of the sketch into small blocks of length $\ell_\mathsf{buf}$, and insert $\mathsf{buf}$ before every block. Note that this only increases the length of the encoding of the sketch by a constant factor. The reason we use such a small block length is that, even if the adversary can forge or destroy some buffers, the total number of bits inserted or deleted caused by this is still small. In fact, we can bound this by $O(k)$ block insertions/deletions with at most $O(k \log n)$ bits inserted/deleted, for which both the sketch and the encoding of the sketch can handle. 
When decoding, we first recognize all the $\mathsf{buf}$'s. 
Then we take the $\ell_\mathsf{buf}$ bits after each $\mathsf{buf}$ to form the decoding of the sketch, and take the remaining bits as the message.

Unfortunately, this approach introduces two additional problems here.\ The first problem is that the original message may contain $\mathsf{buf}$ as a substring. If this happens then in the decoding procedure again we will be taking part of the message to be in the encoding of the sketch.\ The second problem is that the small blocks of the encoding of the sketch may also contain $\mathsf{buf}$. In this case we will be deleting information from the encoding of the sketch, which causes too many edit errors.

To address the first problem, we turn the original message into a \emph{pseudorandom} string by computing the XOR of the message with the output of a pseudorandom generator. Using a $\frac{1}{\poly n}$-almost $ O(\log n)$-wise independence generator with seed length $O(\log n)$, we can ensure that with high probability $\mathsf{buf}$ does not appear as a substring in the XOR. We can then exhaustively search for a fixed seed that satisfies this requirement, and append the seed to the sketch of the document exchange protocol.
 
To address the second problem, we choose the length of the buffer to be longer than the length of each block in the encoding of the sketch, so that $\mathsf{buf}$ doesn't appear as a substring in any block. This is exactly why we choose the length of the buffer to be $\ell_\mathsf{buf}+1$ while we choose the length of each block to be $\ell_\mathsf{buf}$. %We can do this because the length of the buffer and the length of the small block are both $O(\log n)$ and we can let the buffer length to have a larger constant factor.
%To solve the second problem, we further encode each small block into a longer string of length $2 \log n$ with a code $\mathcal{C}_2$ such that $\mathsf{buf}$ doesn't appear as a substring.
%For the $\mathsf{PRG}$, we utilize the $\eps$-almost $\kappa$-wise independence generator.
%For the code $\mathcal{C}_2$, as the lengthes of the input and the codeword are both in $O(\log n)$, we use an exhaustive search for encoding and decoding.

If we directly apply our document exchange protocol to the construction above, we obtain an error correcting code with $O( (k  \log n +t) \log^2 \frac{n}{k\log n + t} )$ redundant bits.
However, by combining the ideas in \cite{CJLW18} and \cite{CormodeM07}, we can achieve redundancy size  $O(k \log n \log \log \log n + t)$, which is better for small $k$  and $t$.

We first briefly describe the construction of the explicit binary code for $k$ edit errors with redundancy $O(k \log n)$ in \cite{CJLW18}. The construction in \cite{CJLW18} starts by observing that a uniform random string satisfies some nice properties.
For example, with high probability, any two substrings of length some $B = O(\log n)$ are distinct.
\cite{CJLW18} calls this property \emph{$B$-distinct}. The construction in \cite{CJLW18} goes by first transforming the message into a pseudorandom string, which is obtained by computing the XOR of the message with an appropriately designed pseudorandom generator. The construction then designs a document exchange protocol for a pseudorandom string with better parameters, and encodes the sketch of the document exchange protocol to give an error correcting code.

The document exchange protocol for a pseudorandom string in \cite{CJLW18} actually consisted of two stages:
in stage I, Alice uses a fixed pattern $p$ to divide her string into blocks of size $\poly(\log n)$. Next, Alice sends a sketch of size $O(k \log n)$ to help Bob recover the partition of her string.
To achieve this, Bob also divides his string into blocks in the same way that Alice does, by using the same pattern $p$.
Alice creates a vector $V$ where each entry of $V$ is indexed by a binary string of length $B$. Specifically, Alice looks at each block in her partition, and stores the $B$-prefix (the prefix of length $B$) of its next block and the length of the current block in the entry of $V$ indexed by the $B$-prefix of the current block.
This ensures each entry of the vector $V$ has only $O(\log n)$ bits. Bob then creates a vector $V'$ in the same way.
\cite{CJLW18} shows that $V$ and $V'$ differ in  at most $O(k)$ entries, thus Alice can send a sketch of size $O(k \log n)$ using the Reed-Solomon Code to help Bob recover $V$ from $V'$. Once this is done, Bob can use $V$ to obtain a guess of Alice's string, which we call $\tilde{x}$, by using his blocks to fill the blocks of $\tilde{x}$, if they have the same $B$-prefix.

Stage II of the construction in \cite{CJLW18} consists of a constant number of levels. In each level, both parties divide each of their blocks evenly into $O(\log^{0.4} n)$ smaller blocks, and Alice generates a sequence of special hash functions called
\emph{$\epsilon$-synchronization hash functions} 
with respect to her string. The nice properties of these hash functions guarantee that in each level Alice can send $O(k \log n)$ bits to Bob, so that Bob can recover all but $O(k)$ blocks of Alice's string. This stage ends in $O(\log_{\log^{0.4} n}( \poly(\log n)) ) = O(1)$ levels when the final block size reduces to $O(\log n)$, at which point Alice can simply send a sketch of size $O(k \log n)$ for Bob to recover her string $x$.

Checking these two stages, it turns out that stage I can be modified to work for block edit errors as well.\ Intuitively, this is because it is still true that such errors won't cause too many different blocks between $V$ and $V'$. On the other hand, stage II becomes problematic, since the use of $\epsilon$-synchronization hash functions crucially relies on the monotone property of standard edit errors. Allowing block transpositions ruins this property, and it is not clear how to give suitable $\epsilon$-synchronization hash functions to work in this case. 

To solve the issue, in stage II, we can apply the 
deterministic document exchange protocol we developed earlier. This implies an error correcting code of redundancy $O(k \log n \log \log n + t)$. However, we show that we can further improve the redundancy to $O(k \log n \log \log \log n + t)$ by using the string parsing idea in \cite{CormodeM07} to improve the partition in Stage I.

Given an input string, string parsing builds a tree where each leaf corresponds to a symbol of the input string, and each non-leaf node corresponds to a substring of the input string.\ Each node of the tree is associated with a label, which is the hash value of its corresponding substring under some hash function.
The structure of the tree only depends locally on the input string, e.g., an edit error on the input string only affects $O(\log n \log^* n)$ nodes of the tree.  

More specifically, string parsing builds the tree bottom-up from one level to another. The labels in the bottom level are obtained by directly applying the hash function to the symbols. Then, the algorithm builds one level of the tree as follows. The labels of the nodes in the previous level form a string of alphabet size $\poly(n)$. The algorithm first finds all repetitive substrings in this string (we say a substring is repetitive, if it's of the form $a^{l}$, for some $l \ge 2$). 
The remaining substrings satisfy the property that any two adjacent symbols are different, and we say such substrings are \emph{non-repetitive}.
\cite{CormodeM07} then applies an algorithm called alphabet reduction to the non-repetitive substrings, 
and obtains a new non-repetitive string for each substring, where the new alphabet is $\{0, 1, 2\}$. 
In particular, the alphabet reduction works in $\log^*n$ steps, where in each step the alphabet size is reduced from the current size $a$ to $\log a$. The reduction keeps doing this until the alphabet size is a constant. Now for all the new strings obtained, the algorithm finds local maximums and local minimums that are not adjacent to any local maximum as \emph{landmarks}, 
and partition the strings into small blocks of length $2$ or $3$ by using the \emph{landmarks} as the boundaries of the blocks. Finally, for each block, the algorithm builds a new node in this level, whose children are the nodes in the block and whose label is the hash value of the subtree.

Here, in our construction of error correcting codes, we use the idea of string parsing in stage I to partition Alice's string $x$ into small blocks. Our goal is to partition the string into blocks of length roughly $\Theta(\log n \cdot \poly(\log \log n))$, while an edit error on the string can only affect a small number of contiguous blocks. In this way, stage II only takes $O(\log \log \log n)$ levels and the sketch size in stage II is $O(k \log n \log \log \log n + t)$.
Note that each node in the parsing tree depends only locally on the input string. We use this property to bound the number of errors among the small blocks obtained in stage I.

More specifically, instead of building a full parsing tree, we only build a partial parsing tree. 
That is, in each level of the parsing tree, we check the number of leaves under each node. 
If a node has more than $T$ leaves for some threshold $T$, we mark the node as `finish'. 
We say a node is a `frozen' node, if all its adjacent nodes are marked as `finish'.
For each `finish' node, we build a new node in the next level, with the only child being this `finish' node.
We then use these `finish' nodes to divide the string into several substrings, and apply the alphabet reduction to the substrings, choose the landmarks, and partition each substring into small blocks according to the landmarks. 
Then for each small block, we build a new node in the next level, and set the children of the new node to be all nodes in the same block.
We keep doing this until each node is either marked as `finish' or  `frozen'.
Finally, we merge each `frozen' node to the `finish' node on its left or right. 
At the end of this process, we obtain several trees, and we partition the string $x$ into small blocks, where each block consists of all the leaves in a tree.
To further improve the parameters and remove the $O(\log^* n)$ factor, we only do two levels of alphabet reduction in each level of the tree.
However, this will result in an alphabet size of $O(\log \log n)$, which means the tree may have $O(\log \log n)$ children.
Hence, the block size may be as large as $O(T \log \log n)$.
Note that each block depends on $O(\log T)$ blocks on its left and right, since in each level of the partial parsing tree, each node depends locally on a constant number of adjacent nodes.
We prove that, if $y$ is obtained from $x$ by $(k, t)$ block edit errors, then the partition of $y$ can be obtained from the partition of $x$ by $(k, O(t / T + k\log T))$ block edit errors over a larger alphabet. If we set $T = \Theta(\log n)$, then in stage I Alice still needs to send a sketch of $O(k \log n \log \log n + t)$ bits. To further reduce the redundancy, we apply the partial parsing tree method again with another threshold $T' = \Theta(\log \log n)$. Now the errors are reduced to $(k, O(\frac{t}{TT'} + k \log T'))$ block edit errors over a larger alphabet, and the block size increases by a $O(T' \log \log n \log\log\log n)$ factor, and becomes $O(\log n (\log \log n)^2)$.%, which is still $O(\log n \cdot \poly(\log \log n))$.

We show that now in stage I, Alice can send a sketch with $O(\frac{t}{TT'} + k\log T')\cdot O(\log n) = O(k \log n\log\log\log n + t)$ bits; and in stage II, Alice can send a sketch with $O(k \log n \log \log \log n + t)$ bits. So the total sketch size is still $O(k \log n \log \log \log n + t)$. By using the asymptotically encoding of Schulman and Zuckerman \cite{796406} and the buffer $\mathsf{buf}$, the final redundancy of the error correcting code is also $O(k \log n \log \log \log n + t)$.

%% file: prelim.tex
\section{Preliminaries} \label{sec:prelim}
\subsection{Notations}

$[n]=\{1,2,\dots,n\}$. Let $\Sigma$ be an alphabet (which can also be a set of strings) and $x\in \Sigma^*$ be a string over alphabet $\Sigma$. $x\in\Sigma^n$ is a string over alphabet $\Sigma$ of length $n$. $|x|$ denotes the length of the string $x$.
Let $x[i,j]$ denote the substring of $x$ from the $i$-th symbol to the $j$-th symbol (Both ends included). Similarly $x[i, j)$ denotes the substring of $x$ from the $i$-th symbol to the $j$-th symbol (not included).
We use $x[i]$  to denote the $i$-th symbol of $x$.
The concatenation of $x$ and $x'$ is $x\circ x'$.
The $B$-prefix of $x$ is the first $B$ symbols of $x$.
$x^N$ is the concatenation of $N$ copies of $x$.
For two sets $A$ and $B$, let $A \Delta B$ denotes the symmetric difference of $A$ and $B$.

Usually we use $U_n$ to denote the uniform distribution over $\{0,1\}^n$.

\subsection{Edit errors}

Consider two strings $x, x'\in \Sigma^*$.

\begin{definition}[Edit distance] The edit distance $ED(x,x')$ is the minimum number of edit operations (insertions and deletions) transforming $x$ to $x'$.
\end{definition}

A subsequence of a string $x$ is a string $x'$ s.t. $x'_{1} = x_{j_1}$, $x'_{2} = x_{j_2}$, $\ldots$,  $x'_{l} = x_{j_l}$, $l = |x'|$, $1\leq j_1 < j_2 < \cdots < j_l \leq |x|$.

\begin{definition}[Longest Common Subsequence] The longest common subsequence between $x$ and $x'$ is the longest subsequence which is the subsequence of both $x$ and $x'$, its length denoted by $LCS(x,x')$.
\end{definition}

We have $ED(x,x') = |x|+|x'|-2LCS(x,x')$.

\begin{definition}[Block-Transposition]
Given a string $x \in \Sigma^n$, the $(i, j, l)$-block-transposition operation for $1 \leq i\leq i + l \leq n$ and $j \in  \{0, \cdots , i-1, i + l , \cdots , n\}$ is defined as an operation which removes $x[i, i+l)$ and inserts  $x[i, i+l)$ right after $x[j]$ in the original string $x$ (if $j=0$, then inserts  $x[i, i+l)$ to the beginning of $x$).

\end{definition}

\begin{definition}[Block edit errors]

A block-insertion/deletion (or burst-insertion/deletion) of $b$ symbols to a string $x$ is defined to be inserting/deleting a block of consecutive $b$ symbols to $x$. When we do not need to specify the number of symbols inserted or deleted, we simply say a block-insertion/deletion.

We define $(k, t)$-block-insertions/deletions (to $x$) to be a sequence of $k$ block-insertions/deletions, where  the total number of symbols inserted/deleted is at most $t$. Similarly, we define $(k, t)$-block edit errors to be a sequence of  $k$ block-insertions, deletions, and transpositions, where the total number of symbols inserted/deleted is at most $t$.

\end{definition}

\subsection{Almost k-wise independence}
\begin{definition}[$\eps$-almost $\kappa$-wise independence in max norm \cite{alon1992simple}]
A series of random variables $X_1, \ldots, X_n \in \mathbb{F}$ are $\eps$-almost $\kappa$-wise independent in Maximum norm if $\forall x \in \mathbb{F}^{\kappa}$, \[\forall i_1, i_2, \ldots, i_{\kappa} \in [n], |\Pr[ (X_{i_1}, X_{i_2}, \ldots, X_{i_{\kappa}}) = x] - 2^{-\kappa} | \leq \eps.\]

A function $g: \{0,1\}^{d} \rightarrow \{0,1\}^{n}$ is an $\eps$-almost $\kappa$-wise independence generator in Maximum norm if $Y = g(U_d) =  Y[1] \circ \cdots \circ Y[n]$ are $\eps$-almost $\kappa$-wise independent in Maximum norm.
\end{definition}
For simplicity,  we neglect the term in Maximum norm when speaking of $\eps$-almost $\kappa$-wise independence, unless specified.

\begin{theorem}[$\eps$-almost $\kappa$-wise independence generator \cite{alon1992simple}]
\label{almostkwiseg}
For every $n, \kappa \in \mathbb{N}$, $\eps > 0$, there exists an explicit $\eps$-almost $\kappa$-wise independence generator $g: \{0,1\}^{d} \rightarrow \{0,1\}^n$, where $d = O(\log \frac{\kappa \log n }{\eps})$.

The construction is highly explicit in the sense that, $\forall i\in [n]$, the $i$-th output bit can be computed in time $\poly(\kappa, \log n, \frac{1}{\eps})$ given the seed and $i$.

\end{theorem}

\subsection{Pseudorandom Generator (PRG)}
\begin{definition}[PRG]
A function $g:\{0,1\}^r\rightarrow \{0,1\}^n$ is a pseudorandom generator (PRG) for a function $f:\{0,1\}^n\rightarrow \{0,1\}$ with error $\eps$ if
\[\left|\Pr[f(U_n)=1]-\Pr[f(g(U_r))=1] \right|\leq \eps\]
where
 $r$ is the seed length of $g$.

Usually this is also called that $g$ $\eps$-fools function $f$. Similarly, if $g$ fools every function in a class $\mathcal{F}$ then we say $g$ $\eps$-fools $\mathcal{F}$ .
\end{definition}

\subsection{Error correcting codes (ECC)}
An ECC $C$ for hamming errors is called an $(n, m ,d)$-code if it has code length $n$, message length $m$, and distance $d$. The rate of the code is defined as $\frac{m}{n}$.

We  utilize the following algebraic geometry codes in our constructions.
\begin{theorem}[\cite{hoholdt1998algebraic}]
\label{agcode}
For every $n, m \in \mathbb{N}, m \leq n, d = n-m-O(1), q = \poly(\frac{n}{d})$, there is an explicit $(n, m, d)$-ECC over $\mathbb{F}_q$ with polynomial-time unique decoding.

Moreover, $\forall n, m \in \mathbb{N}$, for every message $x\in \mathbb{F}_q^{m}$, the codeword is $  x\circ z$ with redundancy $z \in \mathbb{F}_q^{n-m}$.

\end{theorem}
%
%We will utilize systematic Reed-Solomon codes to compute the redundancy for hamming error cases.
%\begin{theorem}[\cite{roth2006introduction}]
%\label{rscode}
%
%There exists an explicit  ECC family  $ \{ (n, m, d)_q\mbox{-code } C \mid n, m \in \mathbb{N}, m \leq n, d = n-m+1, q = \poly(n)\}$ with polynomial-time decoding when the number of errors is less than half of the distance.
%
%
%Moreover, $\forall n, m \in \mathbb{N}$, for every message $x\in \mathbb{F}_q^{m}$, the codeword is $  x\circ z$ with redundancy $z \in \mathbb{F}_q^{n-m}$.
%
%\end{theorem}
For an  ECC $C\subseteq \{0,1\}^n$ for edit errors, with message length $m$, we usually regard it as having an encoding mapping $Enc:\{0,1\}^m\rightarrow \{0,1\}^n$ and a decoding mapping $Dec:\{0,1\}^*\rightarrow \{0,1\}^m\cup\{Fail\}$.

We say an ECC for edit errors is explicit (or has an explicit construction) if
both encoding and decoding can be computed in polynomial time.

To construct ECCs in following sections, we  use an asymptotically good binary ECC for edit errors by Schulman and Zuckerman \cite{796406}.

\begin{theorem}[\cite{796406}]
\label{asympGoodECCforInsdel}
For every $n\in \mathbb{N}$, there is an explicit binary ECC with codeword length $n$, message length $m=\Omega(n)$,  which can correct up to $k_1=\Omega(n)$ edit errors and $k_2=\Omega(n/\log n)$ block-transpositions.

\end{theorem}

%% file: BTprotocol.tex
\section{Deterministic document exchange protocol for block edit errors}
\label{BTprotocol}

\begin{definition}[Collision free hash functions]
\label{def:hash}
Given $n, p, q \in  \mathbb{N}, p \leq n$ and a string $x \in \{0,1\}^n$, we say a hash function $h:\{0,1\}^{p} \rightarrow \{0,1\}^{q}$ is collision free (for $x$), if for every $i,j\in [n - p+1]$, $h(x[i, i+p)) = h(x[j, j+p))$ if and only if $ x[i, i+p) = x[j, j+p) $.

\end{definition}

\begin{theorem}
\label{goodhash}
There exists an algorithm  which, on input $n, p, q \in \mathbb{N}, p \leq n, q = c_0\log n$ for large enough constant $c_0$,   $x\in\{0,1\}^n$, outputs a description of a hash function $h:\{0,1\}^{p} \rightarrow \{0,1\}^{q}$ that is collision free for $x$, in time $\poly(n)$, where the description length is $O(\log n)$.

Also there is an algorithm which, given   the description of  $h$ and any $ u\in \{0,1\}^p$, can output $h(u)$ in time $\poly(n)$.

\end{theorem}

\begin{proof}

Let $\eps = 1/\poly(n)$ be small enough.
Let $g: \{0,1\}^{d} \rightarrow (\{0,1\}^{q})^{ \{0,1\}^p }$ be an $\eps$-almost $2q$-wise independence generator from Theorem \ref{almostkwiseg} with $d = O(\log \frac{2q \log (2^{p}q)}{\eps})$. Here $g$ outputs $q2^p$ bits and we view the output as an array indexed by elements in $\{0,1\}^{p}$, where each entry is in $\{0,1\}^{q}$.

To construct $h$, we try every seed $v \in \{0,1\}^{d}$.
Let $h(\cdot) = g(v)[\cdot]$. This means that, for every $u\in\{0,1\}^p, h(u)$ is the value of the entry indexed by $u$ in $g(v)$.
For any $i, j\in [n-p+1]$, we check whether $h(x[i, i+p)) = h(x[j, j+p))$ if and only if $ x[i, i+p) = x[j, j+p) $. If this is the case then the algorithm returns $h$. The description of $h$ is the corresponding seed $v$.

Now we show that we can indeed find such a $v$ by exhaustive search. If we let $v$ be chosen uniformly randomly, then by a union bound, the probability that there exists $i,j\in [n - p+1]$ s.t. $h(x[i, i+p)) = h(x[j, j+p))$ but $ x[i, i+p) \neq x[j, j+p) $ is at most $1/\poly(n) \cdot n^2 = 1/\poly(n)$. Thus there exists a $v$ s.t. the corresponding $h$ is collision free.

The exhaustive search is in polynomial time because the seed length is $d = O(\log n)$. The evaluation of $h$ is in polynomial time by Theorem \ref{almostkwiseg}. Thus the overall running time of our algorithm is a polynomial in $n$.
\end{proof}

\begin{definition}[Matching]
\label{def:matching}
Given $n, n', p, q \in \mathbb{N}, p\leq n, p\leq n'$, a function $h:\{0,1\}^{p} \rightarrow \{0,1\}^{q}$ and two strings $x\in \{0,1\}^{n}, y\in \{0,1\}^{n'}$, a  matching (may not be monotone) between $x$ and $y$ under $h$ is a sequence of matches (pairs of indices) $w = ((i_1, j_1), \ldots (i_{|w|}, j_{|w|}))$ s.t.
\begin{itemize}
\item for every $k\in [|w|]$,
\begin{itemize}
\item $i_k = 1+ p l_k \in [n]$ for some $l_k$,
\item $j_{k} \in [n']$,
\item $ h(x[i_k, {i_k} + p)) = h(y[j_k ,{j_k} + p))$,
\end{itemize}

\item $i_1, \ldots, i_{|w|}$ are distinct.

\end{itemize}

A non-overlapping matching is a matching with one more restriction.

\begin{itemize}

\item Intervals $ [j_k, j_k+ p ), k\in [|w|]  $, are disjoint.

\end{itemize}

When considering overlaps,  the matching has overlapping degree $d$, if each bit of $y$ appears in at most
$d$ matched pairs for some small number $d$.
\end{definition}
%In the rest of the paper, when we talk about matchings we always mean non-overlapping matchings.
For a match $(i, j)$, it matches two intervals, one from $x$, the other from $y$. When we say  the $y$'s interval (of the match $(i, j)$), we mean $[j , j+p)$, and similarly the $x$'s interval is $[i, i+p)$.
A match $(i, j)$ in a matching is called a wrong match (or wrong pair) if $x[i, i+p) \neq y[j, j+p)$. Otherwise it is called a correct match (or correct pair). A pair of indices $(i,j)$ is called a potential match between $x$ and $y$ if $h(x[i, i+p)) = h(y[j, j+p))$. It may be  wrong because $x[i, i+p)$ may not be $y[j, j+p)$. When $ x, y $ are clear from the context we simply say $(i,j)$ is  a potential match.

To compute a monotone non-overlapping matching we can use the dynamic programming method in \cite{CJLW18}. But our matching is not necessarily monotone. So this raises the question of how hard this problem is.

It seems difficult to find a polynomial algorithm which can exactly compute it. So instead we use constant approximation techniques. There're two difficulties at the first thought. One is that if we compute the non-overlapping matching over the entire strings, then a constant approximation is too bad since there will be $O(n)$ unmatched blocks. So for each level,  we restrict our attention to blocks that are uncovered and wrongly recovered (but discovered by us). The other problem is that we need the approximation rate to be a large enough constant. To achieve this goal, we actually computing matchings with constant degree.

We start from a $ 1/3 $-approximation algorithm, which is greedy.

\begin{construction}

\label{1/3-approxalgo}

Given $n, n', p, q \in \mathbb{N}, p\leq n, p\leq n'$, a polynomial time computable function $h:\{0,1\}^{p} \rightarrow \{0,1\}^{q}$ and two strings $x\in \{0,1\}^{n}, y\in \{0,1\}^{n'}$, we have the following $ 1/3 $-approximation algorithm for computing the non-overlapping matching.

\begin{enumerate}

\item Let the sequence of matches $w$ be empty;

\item Find $i = 1 + pl \in [n]$ and $j \in [n']$,  where $ l \in \mathbb{N}$, s.t.
\begin{itemize}
\item  $h(x[i, i+p)) = h([j, j+p))$,

\item  $i$ is not in any match (as the first entry) of the current $w$,

\item  $[j, j+p)$ does not overlap with any $[j', j'+p)$ for any $j'$  as the second entry in any matches of the current $w$;
\end{itemize}

\item
If there is such a pair of indices $i, j$, then add the match $(i, j)$ to $w$ and go to step 2; Otherwise, output $w$ and stop.

\end{enumerate}

\end{construction}

\begin{lemma}
\label{1/3-approx}

Construction \ref{1/3-approxalgo} gives a $1/3$-approximation algorithm for computing the non-overlapping matching.
\end{lemma}

\begin{proof}

Suppose $w^*$ is the maximum non-overlapping matching between $x, y$ under $h$.

Every time the greedy algorithm adds a  match $(i, j)$  to $w$, we may delete at most $3$ matches in $w^*$. They may be the match which includes $ [i, i+p)$, or the matches whose intervals of $y$ overlap with $ [j, j+p)$.

Note that in the first case, there can be at most $1$ match of $w^*$   deleted since by definition of matching, $ [i, i+p)$ can only be the $x$'s interval for at most $l$ match  of $w^*$. For the second case, note that since $w^*$ is non-overlapping, there are at most two $y$'s intervals, of matches in $w^*$, overlapping with $ [j, j +p)$.

If $|w| < 1/3|w^*|$, then we can delete less than $|w^*|$ matches in $w^*$.

We claim that the matches left can be selected by the greedy algorithm. Suppose one remaining match is $(i, j)$. Note that $ [i, i+p)$ is not in any match of $w$. Since if it is, then this match should have been deleted. Also note that $  [j, j+p) $ does not overlap with any intervals in matches of $w$. Since if it does, then it also should have been deleted.

As a result, if  $|w| < 1/3|w^*|$, our greedy algorithm will not stop. Also note that every time the algorithm conducts step 2 and 3 it will either increase the current matching size by 1, or stop, and the matching size is $O(n/p)$. So our greedy algorithm will halt in polynomial time.

\end{proof}

Next we give an explicit algorithm which computes a even larger matching (better approximation), but it allows overlaps.

\begin{construction}

\label{2/3-approxalgo}

Given $n, n', p, q \in \mathbb{N}, p\leq n, p\leq n'$, a (polynomial time computable) function $h:\{0,1\}^{p} \rightarrow \{0,1\}^{q}$ and two strings $x\in \{0,1\}^{n}, y\in \{0,1\}^{n'}$,
we have the following algorithm.

\begin{enumerate}

\item Let the matching $w  $ be empty, set $S = \{i = 1+ p l \mid l \in \mathbb{N}, i\in [n]\}$, integer $c = 0$;

\item Conduct Construction \ref{1/3-approxalgo} to compute a matching $w'$ between $x_S$ and $y$ under $h$. Here $x_S$ is the projection of x on intervals in set $S$;

\item Let $w = w\cup w'$;

\item Let $S = S \setminus \{ u \mid \exists (u, v) \in w
%,  \mbox{ for some } v \in [n']
\}$;

\item $c= c+1$;

\item If $c \geq 3$, output $w$; Otherwise go to step 2.

\end{enumerate}

\end{construction}

Note that Construction \ref{2/3-approxalgo} is in polynomial time since it simply conducts Construction \ref{1/3-approxalgo} for $3$ times and after each conduction it removes matched blocks of $x$ and only considers the remaining blocks in the next iteration.
So we only need to show its correctness.

\begin{lemma}
\label{2/3-approx}
Construction \ref{2/3-approxalgo} computes a degree $3$ overlapping matching $w$ between $x$ and $y$ under $h$, such that $|w| \geq 2/3 |w^*|$, where $w^*$ is the maximum non-overlapping matching between $x$ and $y$ under $h$.
\end{lemma}

\begin{proof}

Let $w_i, i = 1,2,3$ be the matching the algorithm computes after round $i$. Also let $S_i, i =1,2,3$ be the set $S$ after the $i$th round.

By Lemma \ref{1/3-approx},   $|w_1| \geq 1/3|w^*|$. The number of unmatched blocks is $\bar{n} - |w_1| \leq \bar{n} - 1/3|w^*|$, where $\bar{n} = \lfloor n/p \rfloor$ is the total number of blocks of $x$.

The maximum matching between $x_{S_1}$ and $y$ is at least $|w^*| - |w_1|$. This is because that, each of the matched blocks of $x$ by $w_1$, should be among the $x$'s blocks in the matches of $w^*$. There are at most $|w_1|$ of them. So there are still $|w^*| - |w_1|$ remaining matches in $w^*$ which corresponds to blocks in $x_{S_1}$.

Again by Lemma \ref{1/3-approx}, for $i \geq 2$, at least $1/3(|w^*| - |w_{i-1}|)$ blocks of $x_{S_{i-1}}$ will be matched in the $i$th round.

Thus
\begin{eqnarray}
%\begin{split}
|w_i|&\geq& |w_{i-1}|  + 1/3(|w^*| - |w_{i-1}|)\label{ineq1}\\
	   &=& 1/3|w^*| + 2/3|w_{i-1}| \label{eq1}\\
	   &\geq& (1-(2/3)^{i-1}) |w^*| + (2/3)^{i-1} |w_1|\label{ineq2}\\
	   &\geq& (1-(2/3)^{i-1}) |w^*| + (1/3)(2/3)^{i-1}  |w^*|\label{ineq3}\\
	   &=& (1 - (2/3)^{i})|w^*|\label{eq2}.
%\end{split}
\end{eqnarray}
Inequality \ref{ineq1} is due to Lemma \ref{1/3-approx} as explained above. Equality \ref{eq1} is due to a direct computation. \ref{ineq2} is by recursively applying \ref{ineq1} and \ref{eq1} from $i-1$ to $2$. \ref{ineq3} is because  $|w_1| \geq 1/3|w^*|$.

As a result, $|w_3| \geq 19/27 |w^*| \geq 2/3 |w^*|$.

Note that we apply Construction \ref{1/3-approxalgo} for $3$ times, where in each time, it gives a non-overlapping matching. So each entry of $y$ is in at most one of the matches in that round. So finally we get a degree $3$ overlapping matching.
\end{proof}

We now give the following document exchange protocol.

\begin{construction}
\label{dBTprotocol}

The protocol works for every input length $n\in \mathbb{N}$,  every $(k_1, t)$ block-insertions/deletions  $k_2 $ block-transpositions, $k_1, k_2 \leq \alpha n/\log n, t \leq \beta n$, for some constant $\alpha, \beta$. (If $k_1$ or $k_2 > \alpha n/\log n$, or $t > \beta n$, we simply let Alice send her input string.)
Let $k = k_1 + k_2$.

Both Alice's and Bob's algorithms have $L = O(\log \frac{n}{k\log n + t})$ levels.

For every $i\in [L]$, in the $i$-th level,
\begin{itemize}
\item Let the block size be $b_i =  \frac{n}{18 \cdot 2^{i} (k+\frac{t}{\log n}) }   $, i.e., in each level, divide every block of $x$ in the previous level evenly into two blocks. We choose $L$ properly s.t. $b_L = O(\log n)$;

\item The number of blocks $l_i = n/b_i$;
\end{itemize}

Alice: On input $x \in \{0,1\}^n$,
\begin{enumerate}[label*=\arabic*.]

\item For the $i$-th level,
\begin{enumerate}[label*=\arabic*.]

\item Construct a hash function $h_i: \{0,1\}^{b_i} \rightarrow \{0,1\}^{b^{*} = \Theta(\log n)}$ for $x$ by Theorem \ref{goodhash}.

\item Compute the sequence of hash values i.e. $v[i] = ( h_i(x[1, 1+b_i)), h_i(x[1+b_i, 1+2b_i)), \ldots, h_{i}(x[1+(l_i-1)b_i, l_i b_i)))$;

\item Compute the redundancy $ z[i] \in (\{0,1\}^{b^*})^{\Theta( (k+\frac{t}{\log n}) i)} $ for $v[i]$ by Theorem \ref{agcode}, where the code has distance at least $180(k+ \frac{t}{\log n}) i$;
\end{enumerate}

\item Compute the redundancy $z_{\mathrm{final}} \in (\{0,1\}^{b_L})^{\Theta( (k+\frac{t}{\log n}) \log L )} $ for the blocks of the $L$-th level  by Theorem \ref{agcode}, where the code has distance at least $90(k+ \frac{t}{\log n})L$;

\item Send $h = (h_1, \ldots, h_L)$, $z = (z[1], z[2], \ldots, z[L])$, $v[1]$, $z_{\mathrm{final}}$ to Bob.

\end{enumerate}

Bob: On input $y \in \{0,1\}^{O(n)}$ and received $h, z$, $v[1]$, $z_{\mathrm{final}}$,

\begin{enumerate}[label*=\arabic*.]

\item Create $\tilde{x} \in \{0, 1, *\}^{n}$ (i.e. Bob's current version of Alice's $x$), initiating it to be $(*, *, \ldots, *)$;

\item For the $i$-th level where $1 \leq i \leq L-1$,
\begin{enumerate}[label*=\arabic*.]

%\item Divide $\tilde{x}$ into consecutive blocks to get $\tilde{x}' \in (\{0,1\}^{b_i})^{l_i}$;

\item Apply the decoding of Theorem \ref{agcode} on $ h_i(\tilde{x}'[1, 1+b_i) ),  h_i(\tilde{x}'[1+b_i, 1+2b_i)), \ldots,  h_{i}(\tilde{x}'[1+(l_i-1)b_i, l_i b_i)), z[i]$ to get the sequence of hash values $ v[i] $.  Note that $v[1]$ is received directly, thus Bob does not need to compute it;

\item Let $S = \{j \in [n] \mid h_i(\tilde{x}[1+ (j-1)b_i, 1+ j b_i ) ) \neq v[i][j] \mbox{ or } x[1+ (j-1)b_i, 1+ jb_i ) = (*, \ldots, *) \}$;

\item Compute the matching $w_i = ((p_1, p'_1), \ldots, (p_{|w|}, p'_{|w|})) \in ([l_i] \times [|y|])^{|w_i|}$ between $x_S$ and $y$ under $h_i$, using $v[i]$, by Lemma \ref{2/3-approxalgo};

\item Evaluate $\tilde{x}$ according to the matching, i.e. let  $\tilde{x}[p_j, p_j+b_i) = y[p'_j, p'_j + b_i)$, where $p_j, p_j'\in w_i, j\in [|w_i|]$;

\end{enumerate}

\item In the $L$'th level, apply the decoding of Theorem \ref{agcode} on the blocks of $\tilde{x}$ and $ z_{\mathrm{final}}$ to get $x$;

\item Return $x$.

\end{enumerate}

\end{construction}

\begin{lemma}
\label{dBTprotocolmaxMatchingilen}
For every $i$, the maximum non-overlapping matching between $x_S$ and $y$ under $h_i$ has size at least $ |S|-  (2k_1 + 3k_2 + t/\log n)$.

\end{lemma}

\begin{proof}

%Note that $k_1$ block-insertions/deletions which insert/delete $t$ symbols can delete at most $O(k_1+\frac{t}{\log n})$ blocks of $x$ in the $i$-th level.

Note that block-insertions do not delete bits. One block insertion can corrupt at most one block. For block-deletions, assume that the $j$-th block-deletion delete $t_j$ bits. This can corrupt (delete a block totally or delete part of a block) at most $\lceil t_j/b_i \rceil +1 $   blocks. So the total number of corrupted blocks is at most $ \sum_{j=1}^{k_1}( \lceil t_j/b_i \rceil +1)  \leq  2 k_1 +  t/b_i \leq 2k_1 + t/\log n$.

On the other hand, $k_2$ block-transpositions can corrupt at most $3 k_2$ blocks, because one block-transposition can only corrupt the two blocks at the end of the transposed substring and another block which contains the position that is the destination of the transposition.

As a result, the total number of corrupted blocks is at most $2k_1 + 3k_2 + t/\log n$. After corruption, uncorrupted blocks can be matched to its corresponding blocks (before corruption) in $x$.
So there exists a matching between $x_S$ and $y$ under $h_i$ having size at least $ |S| -  (2k_1 + 3k_2 + t/\log n)$.

\end{proof}

\begin{lemma}
\label{dBTprotocolMatchingLen}
For every $i$, if $v[i]$ is correctly computed by Bob, then  $|w_i| \geq 2/3 (|S| -   (2k_1 +  3k_2 + t/\log n))$.

\end{lemma}
\begin{proof}

By Lemma \ref{dBTprotocolmaxMatchingilen},  the maximum non-overlapping matching between $x_S$ and $y$ under $h_i$ has size at least $ |S|-   (2k_1 +  3k_2 + t/\log n)$. By Lemma \ref{2/3-approx}, $|w_i| \geq 2/3 (|S|-  (2k_1 + 3k_2 + t/\log n))$.

\end{proof}

%\begin{lemma}
%\label{dBTprotocolNWM}
%For every $i$, if $v[i]$ is correctly recovered, then in the $i$-th level the number of wrong matches
%in $w$ is at most $ 3(k + \frac{t}{\log n}) $.
%
%\end{lemma}
%
%\begin{proof}
%
%All wrong matches are between blocks of $x$ and newly created blocks by errors. This is because   $h_i$ is good for $x$. So for any match $(p, p')$ in $w$, if $y[p', p'+b_i)$ is a substring of $x$, then $x[j, j+b_i) = y[p', p'+b_i)$. Thus it is a correct match.
%
%By definition of matching, intervals $y[p_j, p_j+b_i), j\in [|w|]$ are non-overlapping.
%Each block of $y$ in a wrong match must have edit distance at least $1$ from any block of $x$.
%Thus the $j$-th block-insertion of $t_j$ symbols can cause at most $\lceil t_j/b_i\rceil +1$ wrong matches.
%Thus $k_1$ block-insertions can cause at most $\sum_{j=1}^{k_1} (\lceil t_j/b_i + 1 ) \leq  2k_1 + t/b_i$ wrong matches.
%Each block deletion can cause $2$ wrong matches, because it may create two new blocks.
%Also each block-transposition can create at most $3$ wrong matches, one created by removing that block, the other two from inserting it to a new position.
%
%So the total number of wrong matches is at most $ 2 k_1 + t/b_i  + 3k_2  \leq 3(k + t/b_i)$.
%\end{proof}
%

\begin{lemma}
\label{dBTprotocolNWM}
For every $i$, if $v[1], \ldots, v[i ]$ are correctly recovered, then in the $i$-th level the number of wrongly recovered blocks of $x$ is at most $ 3i(2k_1 + 3k_2 + \frac{t}{\log n}) $.

\end{lemma}

\begin{proof}

Consider the matching $w^*$ corresponding to the current recovering of $x$ after $i$ levels, i.e., this matching is generated at level 1 and adjusted level by level. In level $j$, we first use hash values to test every block to see if it is correctly recovered. For wrongly recovered blocks we delete their corresponding matches. Then for remaining wrongly recovered blocks and unrecovered blocks, we compute a matching $w_j$ for them, and add all matches in $w_j$ to $w^*$.

%So $w^*$ is a union of $3i$ non-overlapping matchings between $x$ and $y$ under $h_i$, since every level we use Construction \ref{2/3-approxalgo} to give the matching which is the union of 3 non-overlapping matchings.

For $w_j, j\leq i$, after level $i$, the number of wrongly recovered blocks in level $i$ caused by (the remaining   part of) $w_j$ is at most $ 3( 2k_1 + 3k_2 + \frac{t}{\log n}) $.

This is because in $w_j$ is constructed by Construction \ref{2/3-approxalgo}, which is a union of $3$ matchings.
Each matching of them is non-overlapping. We only need to show that $w_j$, after eliminating detected wrong pairs in these $i$ levels, contains at most $ 2k_1 + 3k_2 + \frac{t}{\log n}$ wrong matches between $x$'s and $y$'s blocks in the $i$-th level. To see this, first note that these matches' $y$ intervals are only from blocks which are modified from $x$'s blocks or newly inserted. For each block-insertion of $t_j$ bits, it can contribute at most $\lceil t_j/b_i \rceil +1 $ wrong matches.   Each block-deletion can contribute at most $2$ wrong matches. So totally block insertions/deletions can cause $ \sum_{j=1}^{k_1}( \lceil t_j/b_i \rceil +1)  \leq  2 k_1 +  t/b_i$ wrong matches.
On the other hand, $k_2$ block-transpositions can contribute at most $3 k_2$ wrong matches, because $1$ block-transposition can only cause $1$ wrong match when deleting the block and inserting the block to its destination may contribute $2$ wrong matches.
Hence the total number wrong matches is at most $2k_1 + 3k_2 + t/b_i$.

Since there are $i$ matchings $w_1, \ldots, w_i$, each containing $3$ non-overlapping matchings, the number of wrongly recovered blocks remaining in $w^*$ is at most $ 3i( 2k_1 + 3k_2 + \frac{t}{\log n}) $.
%
%We show this by using induction on levels.
%
%For the first level, by Lemma \ref{2/3-approx} every entry of $y$ appears in at most 3 intervals in matches of $w_1$. So the total number of wrongly recovered blocks is at most  $ 3 (2k_1 + 3k_2 + \frac{t}{\log n}) $.
%
%Assume the conclusion holds for the first $i-1\geq 1$ level. Consider  level $i$.
%The  $ 3 (i-1)(2k_1 + 3k_2 + \frac{t}{\log n})$ blocks of wrongly recovered blocks in the $(i-1)$-th level will cause at most  $ 6 (i-1)(2k_1 + 3k_2 + \frac{t}{b_{i-1}}) \leq  6 (i-1)(2k_1 + 3k_2 + \frac{t}{b_i}) $ blocks to be wrong in the $i$-th level.

\end{proof}

\begin{lemma}

\label{numofucblks}

For every $i$,  if $v[1], \ldots, v[i ]$ are correctly recovered, then in level $i $, the number of unrecovered blocks is at most $  36 i (k + \frac{t}{\log n})  $.

\end{lemma}

\begin{proof}

We use induction.
%
%For the base case $i=1$, given $v[1]$, by Lemma \ref{dBTprotocolNWM}, the number of wrongly recovered blocks of $x$ is at most
%$  3 (2k_1 + 3k_2 + \frac{t}{\log n})$.
%The number of unrecovered blocks is at most $l_1 - w_1 \leq l_1 - 2/3(l_1 - (2k_1+3k_2 + t/\log n)) \leq 1/3l_1 + 2/3(2k_1 + 3k_2 + t/\log n)$.
%We pick $l_1$ to be small enough s.t. this number is at most $20(k+t/_{b_1})$.
%Each wrongly matched or unrecovered block will be partitioned to be two blocks in the next level, so at level $2$, $|S| \leq 72(k+t/b_{2})$.

For the base case $i=1$, all blocks of $x$ are unknown to Bob. So the number is at most $l_1 = 18\cdot 2^1(k+ t/\log n) = 36(k+t/\log n)$.

For the induction case, assume the number of unrecovered blocks is at most $ 36j(k+t/\log n)$, for all $j\leq i$.
By Lemma \ref{dBTprotocolNWM}, for level $i$ (after the matching is computed), the number of wrongly recovered blocks of $x$ is at most
$$ 3i(2k_1 + 3k_2 + \frac{t}{\log n}) .$$
So at level $i+1$, the number of wrongly recovered blocks is at most doubled, i.e.
$$ | \{ j \in [n] \mid h_{i+1}(\tilde{x}[1+ (j-1)b_{i+1}, 1+ j b_{i+1} )) \neq v[i][j] \} | \leq 6i(2k_1 + 3k_2 + \frac{t}{\log n}) \leq 18i(k+t/\log n).$$

Since Bob has the correct $v[i+1]$, he can detect at most all the wrong blocks.
So $|S|\leq (72+18)i (k+t/\log n) = 90i (k+t/\log n)$.

By Lemma \ref{dBTprotocolMatchingLen}, the number of unrecovered blocks is at most $|S| - |w_i| \leq 1/3|S| + (2k_1+3k_2  +t/\log n) \leq 30(i+1) (k+t/\log n)$.

%
% The number of unrecovered blocks of level $i$ (after matching) is at most
% $$ 2(|S_i| - |w_{i}|) \leq 2/3 |S| + 4/3(2k_1+3k_2 + t/b_{i}) \leq 26 (k + \frac{t}{b_{i+1}}) (i+1).$$
%  So for level $i+1$,
%  $$|S| \leq  6(i+1)(2k_1 + 3k_2 + \frac{t}{b_{i+1}})  +  26 (k + \frac{t}{b_{i+1}}) (i+1) \leq 32 (k + \frac{t}{b_{i+1}})(i+1) .$$

\end{proof}

\begin{lemma}
\label{dBTprotocolcorrect}
Bob can recover $x$ correctly.

\end{lemma}

\begin{proof}

We use induction to show that for every $i\in [L]$, $v[i]$ can be computed correctly by Bob.

For the first level, $v[1]$ is directly received from Alice.

Assume $v[1],\ldots, v[i-1]$ can be computed correctly.
By Lemma \ref{numofucblks}, the number of unrecovered blocks after level $i-1$ is at most $  36(i-1)(k+t/\log n) $. By Lemma \ref{dBTprotocolNWM}, the number of wrongly recovered blocks is at most $ 9(i-1)(k+t/\log n) $.
So the total number of wrongly recovered and unrecovered blocks is at most
$$ 2\times ( 36(i-1)(k+t/\log n) + 9(i-1)(k+t/\log n) )  \leq 90(i-1) (k+t/\log n) <  90i (k+t/\log n).$$
Note that with the redundancy $z[i]$, its corresponding code has distance at least $180(k+t/b_i) i$. So Bob can recover $v[i]$ correctly by Theorem \ref{agcode}.

As a result, at level $L$. By Lemma \ref{dBTprotocolNWM}, the number of wrongly recovered blocks   is at most $ 3L(2k_1 + 3k_2 + \frac{t}{b_L})$. By Lemma \ref{numofucblks} the   number of   unrecovered blocks, is at most $ 36L(k+t/\log n) $.  So the total number of wrongly recovered and unrecovered blocks is at most $45L(k+t/\log n)$. Note that the code distance corresponding to the redundancy $z_{\mathsf{final}}$ is at least $90(k+ t/b_L) L$. So all blocks of $x$ can be recovered correctly by using the decoding from Theorem \ref{agcode}.

\end{proof}

\begin{lemma}
\label{dBTprotocolcc}
The communication complexity is   $ O( (k  \log n +t) \log^2 \frac{n}{k\log n + t} )$.

\end{lemma}

\begin{proof}
For the $i$-th level of Alice, $|z[i]| = \Theta(k+\frac{t}{\log n}) i b^* = \Theta((k \log n+t)i )$. So
$$|z|   = \sum_{i=1}^{L} |z[i]|  = \sum_{i=1}^L O\left( (k \log n+t)i  \right) = O(k\log n+t)L^2.$$

Also $|z_{\mathrm{final}}| = O(k+ \frac{t}{b_L}) \cdot L \cdot O(\log n) = O\left( (k \log n + t)  L \right)$ by Theorem \ref{agcode}.

For every $i\in [L]$, $|h_i| = O(\log n)$ by Theorem \ref{goodhash}. So $|h| = O(\log n)L $.

The length of $v[1]$ is $ l_1 O(\log n) =\frac{n}{b_1} O(\log n) = O( k+ \frac{t}{\log n}  ) \cdot O(\log n)  = O(k \log n + t) $.

Since $L = \log \frac{n}{k\log n + t}$,
  the overall communication complexity is $ O\left( (k  \log n +t) \log^2 \frac{n}{k\log n + t}   \right)$.

\end{proof}

\begin{lemma}
\label{dBTprotocoltc}
Both Alice and Bob's algorithms are in polynomial time.

\end{lemma}

\begin{proof}

For Alice's algorithm, let's consider the $i$-th level. Constructing $h_i$ and evaluating $h_i $  takes polynomial time by Theorem \ref{goodhash}. Computing the redundancy $z[i]$ takes polynomial time by Theorem \ref{agcode}. So the overall running time is polynomial.

For Bob's algorithm, we still consider the $i$-th level. By Theorem \ref{agcode}, getting $v[i]$ takes polynomial time. It takes linear time to visit every block and check if their hash value is equal to the corresponding entry of $v[i]$. By Lemma \ref{2/3-approx}, computing the maximum matching takes polynomial time. So the overall running time is also polynomial.

\end{proof}

\begin{theorem}

\label{dBTdocexc}
There exists an explicit binary   document exchange protocol, having communication complexity  $ O( (k  \log n +t) \log^2 \frac{n}{k\log n + t} )$, time complexity $\poly(n)$, where $n$ is the input size and $k=k_1 + k_2 $, for $(k_1, t)$ block-insertions/deletions and $k_2$ block-transpositions, $   k_1, k_2 \leq \alpha n/\log n, t \leq \beta n$, for some constant $\alpha, \beta$.

\end{theorem}

\begin{proof}
It follows from Construction \ref{dBTprotocol}, Lemma \ref{dBTprotocolcorrect}, Lemma \ref{dBTprotocolcc} and Lemma \ref{dBTprotocoltc}.

\end{proof}

%% file: uniformrandomprotocol.tex
\section{Document exchange for block edit errors of a $B$-distinct string}
\label{randBTProtocol}

\begin{definition}
	For any integer $B$, we say a string $A = A[1], A[2], \cdots, A[n]$ is $B$-distinct, if for any $i \neq j \in [n-B+1]$, $A[i, i+B) \neq A[j, j+B)$.
\end{definition}

\begin{definition}
	We say a string $A = A[1], A[2], \cdots, A[n]$ is a non-repetitive string, if for any $i \in [n-1]$, $A[i] \neq A[i+1]$.
\end{definition}

In this section we prove the following theorem.

\begin{theorem}\label{randmainthm}
	There exists an integer $B=\Theta(\log n)$ such that for any $B$-distinct binary string, there is a polynomial time one way document exchange protocol for $(k, t)$ block edit errors  with communication cost $O(k \log n \log \log \log n + t)$ bits.
\end{theorem}

Recall that, for a integer $B \ge 1$, and a string $A = A[1], A[2], \cdots, A[n]$,  its $B$-prefix is defined as the string $A[1, B]$.
The construction consists of two stages. In Stage I, we partition the string into small blocks. Alice then sends a short sketch to help Bob learn the partition and the $B$-prefix of each block. In Stage II, we modify the Stage II in \cite{CJLW18} to resist block edit errors.
% TECH OVERVIEW 
% P8 down
\subsection{String Partition}
In Stage I, our string partition algorithm uses the string parsing techniques in \cite{CormodeM07}. For completeness, we include \emph{alphabet reduction} and \emph{landmark} in \cite{CormodeM07}.

\medskip
\noindent\textbf{Alphabet reduction}\cite[Sligtly modified]{CormodeM07}

Let $A = A[1], A[2], \cdots, A[n]$ be a string of length $n$, where each $A[i], i \in [n]$ is a symbol in an alphabet $\Sigma$. The alphabet reduction algorithm takes string $A$ as input, 
and outputs a string $A'$ with the same length,
where each symbol $A'[i]$ is computed as follows.
Take two fixed symbols $c_0, c_1$ from $\Sigma$. For the first symbol $A[1]$, if $A[1] = c_0$, set $A[0] = c_1$, otherwise set $A[0] = c_0$.
For each $i \in [n]$, represent $A[i]$ and $A[i-1]$ as binary integers. Let $l$ be the least significant bit in which $A[i]$ and $A[i-1]$ differ. Let $\mathsf{bit}(l, A[i])$ be the $l$-th least significant bit of $A[i]$. Then we define $A'[i] = (l, \mathsf{bit}(l, A[i]))$.

\begin{lemma}\cite[Lemma 1]{CormodeM07}
	For any $i$, if $A[i] \neq A[i+1]$, then $A'[i] \neq A'[i+1]$.
\end{lemma}

Note that the alphabet size of $A'$ is $2 \lceil \log |\Sigma| \rceil$. If we take the alphabet reduction twice, the alphabet size of the resulting string is at most $2 \log\log |\Sigma| + 4$.

\medskip
\noindent\textbf{Landmark}\cite[Slightly modified]{CormodeM07} Let $A$ be a non-repetitive string of length $n$. We take two passes on the string to find the landmarks. 
In the first pass, for each $i \in [3, n-1]$, we say $i$ is a landmark, if $A[i]$ is the local maximum, i.e. $A[i-1] < A[i] > A[i+1]$.
In the second pass, for each $i \in [3, n-1]$, if $A[i]$ is the local minimum, i.e. $A[i-1] > A[i] < A[i+1]$, and $i$ is not adjacent to any landmarks in the first pass, then we say $i$ is a landmark. 

\begin{lemma}\label{lemma:children-num}
Let $A \in \Sigma^*$ be a string. Suppose the landmarks of $A$ are $i_1, i_2, \cdots, i_{n'}$.
If we partition the string $A$ as $A[1, i_1), A[i_1, i_2), \cdots, A[i_{n'}, n]$.
 then the length of each substring is in the range of $[2, |\Sigma|+1]$.
\end{lemma}

\begin{proof}
The substring between any two adjacent landmarks must be monotone. Hence, for any two adjacent landmarks, we have $1 < |i - j| < |\Sigma|$. For the first and the last substring, their lengthes are at most $|\Sigma|+1$.
\end{proof}

\begin{construction}[Algorithm : $\mathsf{Partition}$]
Input : A threshold integer $T$, and a $1$-distinct string $x$ of length $n$ over alphabet $\Sigma$.

Output : A series of $n'$ indices $(i_0 = 1 < i_1 < i_2 < \cdots < i_{n'} = n+1)$, which corresponds to the following partition of $x$ : $x[i_0, i_1), x[i_1, i_2), \cdots, x[i_{n'-1}, i_{n'})$.

The algorithm builds a series of trees, where each node is associated with a label in $\Sigma$.
We finally output the indices corresponding to the roots of the trees.

The algorithm builds the trees level by level.
Initially, each position of the input string corresponds to a single leaf node.
In each level, the algorithm partition the nodes in the current level into blocks, and create a new node for each block in the next level, where the children of the new node are set to be the the nodes in the block.

Let $n_0 = n$, $x_0 = x$, and $i_0^{(0)} = 1, i_1^{(0)} = 2, \cdots, i_{n}^{(0)} = n+1$.

For the $h$-th level, there are $n_h$ nodes, the labels on these nodes form a string $x_h \in \Sigma^{n_h}$. 
Each node has some leaves, and these leaves form a contiguous interval in $x$. We denote the leaves interval of $j$-th node in $h$-th level as $[i_{j-1}^{(h)}, i_j^{(h)})$, then the label of $j$-th node is $x_h[j] = x[i_{j-1}^{(h)}] \in \Sigma$.
We apply the alphabet reductions to $x_h$ and partition $x_h$ according to the \emph{landmarks}, and thus obtain $x_{h+1}$ and $\{i^{(h+1)}_{j}\}_{j}$ for the next level.

Now for each $h = 0, 1, 2, \cdots, \lceil \log T \rceil$, we do the following steps:
\begin{enumerate}
	\item For each $j \in [n_h+1]$, consider $j$-th node in $h$-th level. If $i_{j}^{(h)} - i_{j-1}^{(h)} \ge T$, mark the $j$-th node as `finish'.
	If there are no adjacent non-`finish' nodes, then we merge each non-`finish' node to the `finish' node to its left or right. That is, output $\{i_j^{h} \mid \text{ j-th node is a `finish' node.} \}\cup \{n+1\}$ and halt.
	
	\item 
	Use the `finish' nodes to partition the indices of string $x_h$ into intervals $I_1, I_2, \cdots$, where each substring $x[I_1], x[I_2], \cdots$ doesn't contain any `finish' node.
	Apply alphabet reduction \emph{twice} to the substrings $x_h[I_1], x_h[I_2], \cdots$, and obtain $x_h'[I_1], x_h'[I_2], \cdots$. Further partition the intervals $I_1, I_2, \cdots$ into small blocks by the \emph{landmarks} in $x_h'[I_1], x_h'[I_2], \cdots$.
	
	\item 
	For each node marked as `finish', we build a new node in the next level. The `finish' node is the single child of the new node. 
	Next, we iterate on all blocks in string $x_h$.
	For $j$-th block $x_h[l, r)$, 
	we create a new node in the next level.
	The children of the new node are all nodes in $x_h[l, r)$. Now the label of new node is $x_{h+1}[j] = x_{h}[l]$, and the range of the leaves of the new node is $[i_{l-1}^{(h)}, i_{r-1}^{(h)})$. We set this range as $[i_{j-1}^{(h+1)}, i_{j}^{(h+1)})$.
\end{enumerate}
\end{construction}

\begin{figure}[H]
	\centering
	\includegraphics[width=0.9\linewidth]{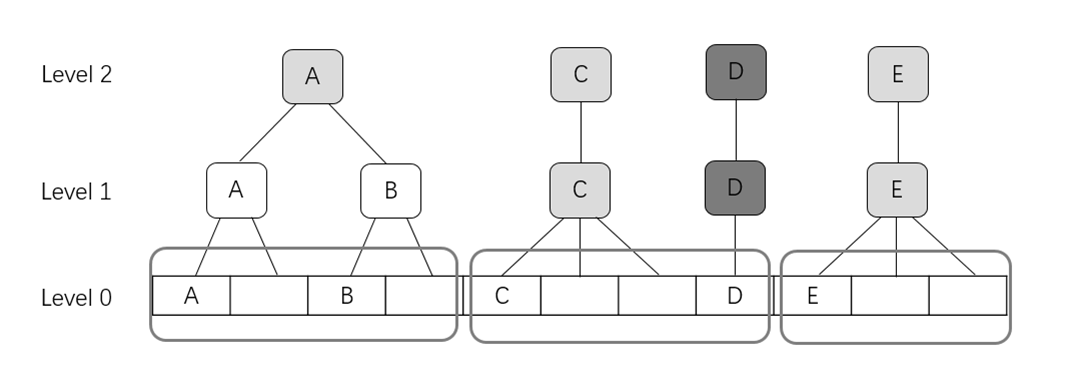}
	\caption{An example of $\mathsf{Partition}$, where the input string is at the bottom level, the nodes in light grey are `finish' nodes, and the nodes in deep grey are `frozen' nodes.}
\end{figure}

\begin{lemma}\label{lemma:block-size}
	Let $T$ be a threshold parameter,
	and $x$ be a $1$-distinct string in an alphabet of size $O(\log n)$, then $\mathsf{Partition}(T, x)$ stops in $\lceil \log T \rceil$ levels. 
\end{lemma}

\begin{proof}
For any $h$, we say $j$-th symbol of $x_h$ is frozen, if both the $(j-1)$-th symbol and the $(j+1)$-th symbol are marked as `finish'.
%TODO : for boundry?
To prove that the algorithm stops in 
$\lceil \log T \rceil$ levels, we first prove the following claim: 
for the $j$-th symbol in level $h$, if it's not marked as `finish' or frozen, then $i_j^{(h)} - i_{j-1}^{(h)} \ge 2^h$.

We prove the claim by induction. 
The claim is true for the $0$-th level, since $i_{j}^{(0)} - i_{j-1}^{(h)} = 1 \ge 2^0$. 
Now let's assume the claim is true for the $h$-th level, and the goal is to prove the claim for the $(h+1)$-th level. 
For each block in the $(h+1)$-th level, if it's neither marked as `finish' nor frozen,
from Lemma~\ref{lemma:children-num}, the block must have at least two children, 
and the length of the block is at least $2^h + 2^h \ge 2^{h+1}$. Hence, the claim is true for all levels.

Now we prove the lemma by contradiction. If the algorithm doesn't stop in the $\lceil \log T \rceil$-th level, then there exists a block in level $\lceil \log T \rceil$ such that it's not marked as `finish' or frozen. Thus the length of the block is at least $2^{\lceil \log T \rceil} \ge T$, hence it should be marked as `finish', which is a contradiction.
\end{proof}

\begin{lemma}\label{lemma:block-local}
	Let $T$ be a threshold parameter,
	and $x$ be a $1$-distinct string over an alphabet $\Sigma$.
	Let $\{i_0=1, i_1, i_2, \cdots, i_{n'} = n+1\} = \mathsf{Partition}(T, x)$, then each block $[i_{j-1}, i_j), j = 1, 2, \cdots, n'$ depends on $O(\log T)$ blocks on its left, and $O(\log T)$ blocks on its right. Moreover,  $T \le i_{j}-i_{j-1} \le T\cdot (2 \log \log |\Sigma| + 7)$.
\end{lemma}

\begin{proof}
	We first prove the size of each block is in the range of $[T, T\cdot (2 \log \log |\Sigma| + 7)]$. Since we do the alphabet reduction twice, from Lemma~\ref{lemma:children-num}, each node has at most $2\log\log |\Sigma| + 5$ children. Hence, the size of each `finish' block is bounded by $T\cdot(2\log\log |\Sigma| + 5)$. As we may add two `frozen' nodes to this `finish' node in the last level, and the sizes of them are bounded by $T$, the total length of the block is bounded by $T\cdot(2\log\log |\Sigma| + 5)+2T = T\cdot(2\log \log |\Sigma| + 7)$.
	
	For any $h = 0, 1, 2, \cdots, \lceil \log T \rceil$,
	in the $h$-th level, we claim that the $j$-th node $x_h[j]$ depends on $l_h$ blocks on its left, and $r_h$ blocks on its right, where $l_h = 100 h$, and $r_h = 100 h$.
	
	We prove the claim by induction on $h$. 
	For $h = 0$, the claim holds. Now we assume the claim holds for level $h$, and we prove that the claim holds for level $(h+1)$.
	For each $j = 1, 2, \cdots, n_{h+1}$, consider the $j$-th node in level $(h+1)$. There are two cases: (1). the $j$-th node is obtained by some `finish' $j'$-th node in $h$-th level, and (2). the $j$-th node has children from some $j_1'$-th to $j_2'$-th nodes in $h$-th level.
	
	For the first case, the $j$-th node in $(h+1)$-th level also depends on $l_h$ blocks on its left and $r_h$ blocks on its right, so the claim is true for $(h+1)$-th level.
	For the second case, the $j$-th nodes depends on at most $3$ blocks on the left of $j_1'$, and at most $4$ blocks on the right of $j_2'$. Note that $j_1', j_2'$ and $j$ are in the same block.
	Hence, the $j$-th nodes depends on $l_{h+1} = l_h + 3$ blocks on its left, and $r_{h+1} = r_h + 4$ blocks on its right. Hence, the claim holds for the $(h+1)$-th level.
	
	From Lemma~\ref{lemma:block-size}, there are at most $\lceil \log T \rceil$ levels, so we finish the proof.
\end{proof}

\begin{lemma}\label{lemma:partition-reduce}
	Let $T$ be a threshold parameter,
	and $x$ be a non-repetitive string over an alphabet $\Sigma$. Let $\{i_0, i_1, i_2, \cdots\} = \mathsf{Partition}(T, x)$.
	Suppose $y$ is the string obtained by applying $(k, t)$ block edit errors to $x$.
	Let $\{i_0', i_1', i_2', \cdots\} = \mathsf{Partition}(T, y)$. Then string $x[i_0]\circ x[i_1] \circ x[i_2] \circ \cdots$ and $y[i_0'] \circ y[i_1'] \circ y[i_2'] \circ \cdots$ differ by at most $(k, O(t /T + k \log T))$ block edit errors.
\end{lemma}

\begin{proof}
	Since any $(k, t)$ block edit errors can be regarded as a series of $k$ single block edit error with parameters $(1, t_1), (1, t_2), \cdots, (1, t_k)$, where $\sum_{i \in [k]} t_i = t$. For each block edit error with parameter $(1, t_i), i \in [k]$, there are two cases: (1). the block edit error is an insertion or deletion of $t_i$ contiguous symbols, (2).$t_i = 0$, and the block edit error is a block transposition moving symbols in $[j, j')$ to the position $j''$.
	
	For case (1), w.l.o.g, we only need to prove for a block insertion of length $t$. From Lemma~\ref{lemma:block-local}, each block depends on $O(\log T)$ neighboring blocks. 
	Hence, $t_i$ contiguous insertion error affects at most contiguous $O(t_i / T + \log T)$ blocks.
	For the second case, from Lemma~\ref{lemma:block-local}, a block transposition error affact at most $O(\log T)$ blocks near the indices $j, j'$ and $j''$. Hence, the total number of blocks affected in this case is still $O(\log T)$.
	Summing up the number of all affacted blocks, we bound the number of affected blocks by $O(\sum_{i \in [k]} (t_i/T + \log T) ) = O(t/T + k \log T)$.
\end{proof}

\subsection{Document exchange protocol}
As stated before, our document exchange protocol for a $B$-distinctive binary string has two stages.
Stage I is modified from the stage I in the construction of \cite{CJLW18} and combined the idea of parsing tree in \cite{CormodeM07}, to resist block edit errors.

\begin{construction}[Stage I, modified from \cite{CJLW18}]\label{firststage}
    Let $n$ denote the length of Alice's string $x$, $T= B = 3 \log n, T' = \log T = \Theta(\log \log n), T'' = T T' (\log\log n)^2 \log\log\log n = \Theta(\log n (\log \log n)^3 \log \log \log n)$. \\ %For simplicity of description, we assume $n$ is a multiple of $B$. \\

Alice: On input a $B$-distinct string $x \in \{0, 1\}^n$.
\begin{enumerate}[label*=\arabic*]
	\item Create a string $\bar{x}$ of length $\bar{n} = n-B+1$, each symbol of $\bar{x}$ is an element in $\{0, 1\}^B$. Let $\bar{x} = x[1, B], x[2, B+1], \cdots, x[n-B+1, n]$.
	
	\item Compute a partition of $\bar{x}$: $\{i_0 = 1, i_1, i_2, \cdots, i_{n'} = \bar{n}+1\} = \mathsf{Partition}(T, \bar{x})$. 
	Create a string $x'$ with alphabet $\{0, 1\}^B$ : $x' = \bar{x}[i_0], \bar{x}[i_1], \bar{x}[i_2], \cdots, \bar{x}[i_{n'}]$. 
	Now apply $\mathsf{Partition}$ to $x'$ again and obtains a partition of $x'$: $\{i_0'=1, i_1', i_2', \cdots, i_{n''}' = n'+1\} = \mathsf{Partition}(T', x')$.
	Combine the two partitions and obtain the following partition on $\bar{x}$ : $I''=\{i_{i'_{j-1}-1} \mid j \in [n''+1]\}$. Denote $I'' = \{i''_0, i''_1, i''_2, \dots, i''_{n''}\}$, where $i''_0 = 1 < i''_1 < i''_2 < \dots < i''_{n''} = \bar{n}+1$.
	Finally partition $x$ into the blocks $x[i''_0, i''_1), \cdots, x[i''_{n''-1}, i''_{n''})$.
	
	\label{AlicepartitionI}
	
	\item Create a set $V = \left\{(\textsf{len}_b, \textsf{B-prefix}_b, \textsf{B-prefix}_{b+1}) \mid 1 \le b \le n'' - 1\right\}$,
	where $\textsf{len}_b$ is the length of the $b$-th block, and $\textsf{B-prefix}_b$ and
	$\textsf{B-prefix}_{b+1}$ are the B-prefix of the $b$-th block and
	the $(b+1)$-th block respectively.
	
	\item Represent the set $V$ as its indicator vector, which has size $\poly(n)$,
        and send the redundancy $z_V$ being able to correct $\Theta(k \log \log \log n + t / TT')$ Hamming errors, using Theorem~\ref{agcode} (or simply using a Reed-Solomon code).
	
	\item Partition the string $x$ evenly into $n / T''$ blocks, each of size $T''$. \label{AlicelastI}
\end{enumerate}

Bob: On the redundancy $z_V$ sent by Alice, 
and the string $y$ obtained from $x$ by $(k, t)$ block edit errors.

\begin{enumerate}[label*=\arabic*]
	\item Create a string $\bar{y}$ of length $\bar{m} = m-B+1$, each symbol of $\bar{y}$ is in $\{0, 1\}^B$. Let $\bar{y} = y[1, B], y[2, B+1], \cdots, y[n-B+1, n]$.
	
	\item Compute a partition of $\bar{y}$: $\{i_0 = 0, i_1, i_2, \cdots i_{m'} = \bar{m}+1\} = \mathsf{Partition}(T, y)$. 
	Create a string $y'$ with alphabet $\{0, 1\}^B$ : $y' = \bar{y}[i_0], \bar{y}[i_1], \bar{y}[i_2], \cdots, \bar{y}[i_{m'}]$.
	Now apply $\mathsf{Partition}$ on $y'$ and obtain a partition of $y'$: $\{i_0'=1, i_1', i_2', \cdots i_{m''}' = m'+1\} = \mathsf{Partition}(T', y')$.
	Combine the two partitions and obtain the following partition on $\bar{y}$ : $I''=\{i_{i'_{j-1}-1} \mid j \in [m''+1]\}$. Denote $I'' = \{i''_0, i''_1, i''_2, \dots, i''_{m''}\}$, where $i''_0 = 1 < i''_1 < i''_2 < \dots < i''_{m''} = \bar{m}+1$.
	Finally partition $y$ into the blocks $y[i''_0, i''_1), \cdots, y[i''_{m''-1}, i''_{m''})$.
	\label{Bob1}
	
	\item Create a set $V' = \left\{(\textsf{len}_b, \textsf{B-prefix}_b, \textsf{B-prefix}_{b+1}) \mid 1 \le b \le m'' - 1\right\}$
	using the partition of $y$.
	\label{Bob2}
	
	\item Use the indicator vector of $V'$ and the redundancy $z_V$ to recover Alice's set $V$.
	
	\item Create an empty string $\tilde{x}$ of length $n$,
	and partition $\tilde{x}$ according to the set $V$ in the following way:
	first find the element $(\textsf{len}^{(1)}, \textsf{B-prefix}^{(1)}, \textsf{B-prefix}'^{(1)})$
	in $V$ such that
	for all elements $(\textsf{len}, \textsf{B-prefix}, \textsf{B-prefix}')$ in $V$,
	$\textsf{B-prefix}^{(1)} \neq \textsf{B-prefix}'$.
	Then partition $\tilde{x}[1, \textsf{len}^{(1)}]$ as the first block,
	and fill $\tilde{x}[1, B]$ with $\textsf{B-prefix}^{(1)}$. Then find the element $(\textsf{len}^{(2)}, \textsf{B-prefix}^{(2)}, \textsf{B-prefix}'^{(2)})$
	such that $\textsf{B-prefix}^{(2)} = \textsf{B-prefix}'^{(1)}$,
	and partition $\tilde{x}[\textsf{len}^{(1)}+1, \textsf{len}^{(1)} + \textsf{len}^{(2)}]$ as the second block, and fill $\tilde{x}[\textsf{len}^{(1)}+1, \textsf{len}^{(1)} + B]$ with $\textsf{B-prefix}^{(2)}$.
	Continue doing this until all elements in $V$ are used to recover the partition of $x$.
	
	\item For each block $b$ in $\tilde{x}$,
	if Bob finds a unique block $b'$ in $y$ such that the B-prefix of $b'$ matches the B-prefix of $b$ and the lengths of $b$ and $b'$ are equal,
	Bob fills the block $b$ using $b'$.
	If such $b'$ doesn't exist or Bob has multiple choices of $b'$,
	then Bob just leaves the block $b$ as blank.
	\label{Bob5}
	
	\item Partition the string $\tilde{x}$ evenly into $n / T''$ blocks, each of size $T''$. \label{BoblastI}
\end{enumerate}	
\end{construction}

\begin{construction}[Stage II]\label{secondstage}
	Stage II consists of $O(\log \log \log n)$ levels. 
	
	Let $L = O(\log n)$,   $i^* =L - O(\log \log \log n)$ be s.t. $b_{i^*} \geq T'' \geq b_{i^*+1}$ where $b_{i } = \Theta(\frac{n}{ 2^{i } (k+\frac{t}{\log n}) } )$ for each $i \in [L]$, and $b_L = O(\log n)$.
	
Alice does the following.
\begin{enumerate}

\item For $i = i^*$ to $L-1$,

\begin{itemize}

\item Construct a hash function $h_{i}: \{0,1\}^{b_{i}} \rightarrow \{0,1\}^{B}$ for $x$ by using the first $B$ bits of the input as the output.

\item Compute the sequence $v[{i}] = ( h_{i}(x[1, 1+b_{i})), h_{i}(x[1+b_{i}, 1+2b_{i})), \ldots, h_{{i}}(x[1+(l_{i}-1)b_{i}, l_{i} b_{i})))$;

\item Compute the redundancy $ z[{i}] \in (\{0,1\}^{B})^{\Theta(k+\frac{t}{b_{i}})} $ for $v[{i}]$ by Theorem \ref{agcode}, where the code has distance $c^*(k+ \frac{t}{b_i})$ with $c^*$ being a large enough constant;	

\end{itemize}

%\begin{itemize}
%
%\item Construct a hash function $h_{i^*}: \{0,1\}^{b_{i^*}} \rightarrow \{0,1\}^{B}$ for $x$ by using the first $B$ bits of the input as the output.
%
%
%\item Compute the sequence $v[{i^*}] = ( h_{i^*}(x[1, 1+b_{i^*})), h_{i^*}(x[1+b_{i^*}, 1+2b_{i^*})), \ldots, h_{{i^*}}(x[1+(l_{i^*}-1)b_{i^*}, l_{i^*} b_{i^*})))$;
%
%
%\item Compute the redundancy $ z[{i^*}] \in (\{0,1\}^{b^*})^{\Theta(k+\frac{t}{b_{i^*}})} $ for $v[{i^*}]$ by Theorem \ref{agcode}, where the code has distance $c^*(k+ \frac{t}{B})$ with $c^*$ being a large enough constant;	
%
%\end{itemize}

\item Compute $z_{\rm{final}}$ which is the redundancy for $ x[1, 1+ b_L), \ldots, x[1+(l_L -1)b_L, n)$ by Theorem \ref{agcode}, where the code has distance $ c_{\rm{final}}( k+ t/b_L)$ with $c_{\rm{final}}$ being a large enough constant.

\item Send 
%$h = (h_{i^*}, h_{{i^*}+1}, \ldots, h_{L^*}),$
$ z[{i^*}], z[{i^*}+1], \ldots, z[L ], z_{\rm{final}}$.
\end{enumerate}

Bob conducts the following. Assume now his version of $x$ is $\tilde{x}$ (which is the input for this stage).

\begin{enumerate}

\item For $i = i^*$ to $L$,
\begin{itemize}
\item Apply the decoding of Theorem \ref{agcode} on $ h_{i}(\tilde{x}'[1, 1+b_{i}) ),  h_{i}(\tilde{x}'[1+b_{i}, 1+2b_{i})), \ldots,  h_{{i}}(\tilde{x}'[1+(l_{i}-1)b_{i}, l_{i} b_{i})), z[{i}]$ to get the sequence of hash values $ v[{i}] $;

\item Compute the matching $w_i = ((p_1, p'_1), \ldots, (p_{|w_i|}, p'_{|w_i|})) \in ([l_{i}] \times [|y|])^{|w_i|}$  between $x$ and $y$ under $h_{i}$, using $v[{i}]$, in the following way:
\begin{itemize}
\item Mark every symbol of $y$ as unused and let $w_i$ be empty; 

\item Consider every $j\in [l_i]$. Find $p'_i $ which is the smallest index in $[|y|]$ s.t. $h_i(y[p'_j, p'_j+b_i)) = v[i][j]$ and $y[p'_j, p'_j+b_i)$ only contains unused symbols. If there is such $p'_j$, then add $(p_j, p'_j)$ to $w_i$ and mark every symbol in $y[p'_j, p'_j+b_i)$ as used; 

\item return $w_i$;
\end{itemize}

\item Evaluate $\tilde{x}$ according to the matching, i.e. let  $\tilde{x}[p_j, p_j+b_{i}) = y[p'_j, p'_j + b_{i}), j\in [l_i]$;
\end{itemize}

\item Apply the decoding of Theorem \ref{agcode} on the blocks of $\tilde{x}$ and $ z_{\mathrm{final}}$ to get $x$;

\end{enumerate}

\end{construction}

\subsection{Analysis}

\begin{lemma}\label{setdiffbnd}
    If Alice's input string $x$ is $B$-distinct, then $|V \Delta V'| \le O(k \log \log \log n + t / TT')$.
\end{lemma}

\begin{proof}
In the first step of Alice and Bob, the string $\bar{x}$ and $\bar{y}$ differ by at most $(k, t + k \log n)$ block edit errors.
From Lemma~\ref{lemma:partition-reduce}, the string $x'$ and $y'$ differ by at most $(k, O(t + k \log T)/T + k \log T = O(t/T + k\log T))$ block edit errors.
Applying Lemma~\ref{lemma:partition-reduce} again, 
we derive that $x''$ and $y''$ differ by at most $(k, t'')$ block edit errors, where
\begin{align*}
t'' = O\left(\frac{t/T + k\log T}{T'} + k \log T'\right) = O(k \log \log \log n + t / TT')
\end{align*}

Since string $x$ is $B$-distinct, the symbols in string $\bar{x}$ are $1$-distinct, so are the symbols in $x''$. $(k, t'')$ block edit errors can affect at most $O(k + t'')$ elements in $V$. Hence, $|V \Delta V'| \le O(k + t'') = O(k \log \log \log n + t / TT')$.
\end{proof}

\begin{theorem}\label{stageIthm}
If Alice's input string $x$ is $B$-distinct, then after Stage I, at most $O(k + t / T'')$ blocks of $\tilde{x}$ contains unfilled bits or incorrectly filled bits.
\end{theorem}

\begin{proof}
%In this proof, we say a block is a \emph{bad} block, if it contains unfilled bits or incorrectly filled bits.
By Lemma~\ref{setdiffbnd}, Bob can recover the set $V$ correctly using $z_{V}$.
We say a block in Bob's step~\ref{BoblastI} is a \emph{bad block}, if it contains unfilled bits or incorrectly filled bits.

Suppose the $(k, t)$ block edit errors are a series of $k$ single block edit error with parameters $(1, t_1), (1, t_2), \\ \cdots, (1, t_k)$, where $\sum_{i \in [k]} t_i = t$.
For each of block edit error with parameter $(1, t_i), i \in [k]$, 
there are two cases: (1). the $i$-th error is a block insertion or deletion. (2). the $i$-th error is a block transposition moving the substring $[j, j')$ to the position $j''$.

We first prove that each block in Bob's step \ref{Bob5} depends on $O(T'')$ symbols of $\bar{x}$ on its left and its right.
From Lemma~\ref{lemma:block-local}, each symbol of $x'$ depends on $O(\log T) = O(\log \log n)$ neighboring blocks,
and from Lemma~\ref{lemma:block-size}, each blocks contains at most $O(T \log \log n) = O(\log n \log \log n)$ indices of $\bar{x}$. Hence, each symbol of $x'$ depends on $O(\log \log n) \cdot O(\log n \log \log n) = O(\log n \cdot (\log \log n)^2)$ contiguous symbols of $\bar{x}$. Similarly, each symbol of $x''$ depends on $O(\log T') = O(\log \log \log n)$ neighboring blocks, and each block contains most $O(T' \log \log n) = O((\log \log n)^2)$ indices of $x'$. Hence, each symbol of $x''$ depends on $O((\log \log n)^2)\cdot O(\log \log \log n) = O((\log \log n)^2 \log \log \log n)$ contiguous symbols of $x'$. 
Now we can conclude that each symbol of $x''$ depends on $O(((\log \log n)^2 \log \log \log n) + O(\log T))\cdot O(T \log \log n) = O(\log n (\log \log n)^3 \log \log \log n) = O(T'')$ symbols of $\bar{x}$ on its left and right.

For case (1), as each block has size at least $T''$, from the argument above, inserting or deleting a block of size $t_i$ can affect at most $(t_i + O(T'')) / T'' = O(t_i / T'') + O(1)$ blocks. For case (2), from the argument above, we derive that the block transposition can only affect $O(T'') / T'' = O(1)$ blocks in Bob's partition of $y$. Hence, the number of bad blocks created by this error is at most $O(1)$. 

We finish the proof by summing up all bad blocks in $k$ errors.
\end{proof}

\begin{proof}[Proof of Theorem~\ref{randmainthm}]
%	From Theorem~\ref{Bdistinctandinterval}, with probability $1 - 1/\poly(n)$, all properties in Theorem~\ref{Bdistinctandinterval} hold.

We use induction to show the claim that at level $L$, the number of unfilled blocks or wrongly filled blocks is at most $c' (k + t/b_{L})$ for some constant $c'$.

For level $i^*$, by Theorem~\ref{stageIthm} there are at most $O(k+ t/T'') = O(k+t/B)$ blocks, each having length $T''$, that contain uncovered bits or incorrectly recovered bits. 
Let $c'$ be the maximum of constant factor here and the number $2$.

Assume for level $i-1 \in [i^*, L)$, our claim holds. In Construction \ref{secondstage}, since $c^*$ is a large enough constant, by Theorem \ref{agcode} $v[i]$ can be recovered correctly by Bob.  
Consider the computing of $w_i$ using $v[i]$, $y$. Note that $1$ block insertion of $a$ bits can cause at most $2 + a/b_{i}$ wrongly filled blocks or unfilled blocks. So
 $k_1$ block insertions/deletions of $t$ bits can create at most $2 k_1 + t/b_{i} $ wrongly filled blocks or unfilled blocks. Also note that one block transposition  can cause at most $ 2 $  wrongly filled blocks or unfilled blocks. So $k_2$ block transpositions can cause at most $2k_2 $ wrongly filled blocks or unfilled blocks. So the total number of wrongly filled or unfilled blocks in $i$ level is at most $2k + t/b_{i} \leq c'(k+ t/b_{i})$.
 
This shows our claim.  Note that by this claim, also since $c_{\rm{final}}$ is a large enough constant, Bob can recover all blocks of $x$ in level $L$ correctly using $z_{\mathrm{final}}$. 

Next we compute the communication complexity.

For stage I,
	the size of $z_V$ is $ O(k \log \log \log n + \frac{t}{TT'})\log n  = O(k \log \log \log n + t)$ bits.
	%From Theorem~\ref{stageIthm}, the first level of Stage II sends $O(k + t/B) \cdot O(\log n) = O(k \log n + t)$ bits.

%The following statements (to the end of the proof) are what I added, you can change arbitrarily as you want.
	For stage II,
	note that since $b_{i^*} \geq T \geq b_{i^*+1}$, $i^* = L^* - O(\log \log \log n)$. For every level $i$, $|z[i]| = O(k+ t/b_i)B$.
	So 
	$$\sum_{i=i^*}^L |z[i]|  = \sum_{i=i^*}^{L} O(k+ t/b_i)B =  O(k \log n \log \log \log n + t).$$
Also note that $|z_{\rm{final}}| = O(k+ t/b_L)B$. Thus the overall communication cost is $O(k \log n \log \log \log n + t)$.

%	What remains to prove is that after the extra level, $\tilde{x}$ has at least $l_{i^*+1} - 12(k+ \frac{t}{b_{i^*+1}})$ blocks the same as that of $x$ (in order to use Corollary \ref{dBTdocexcCorollary}), where the block length is $b_{i^*+1}$. Before the extra level $\tilde{x}$ 	 has at least $l_{i^*} - O(k+ \frac{t}{B})$ blocks the same as that of $x$ (i.e. being recovered). Since $z[i^*]$ is the redundancy of a code which has distance $c^*(k+ t/B)$ with $c^*$ being large enough, $v[i^*]$ can be recovered by Bob using $ h_{i^*}(\tilde{x}'[1, 1+b_{i^*}) ),  h_{i^*}(\tilde{x}'[1+b_{i^*}, 1+2b_{i^*})), \ldots,  h_{{i^*}}(\tilde{x}'[1+(l_{i^*}-1)b_{i^*}, l_{i^*} b_{i^*})), z[{i^*}]$, by Theorem \ref{agcode}. Note that Lemma \ref{SDPmatchingalgo} guarantees that $w^*_{i^*}$ is the maximum matching computed using $v[i^*]$ and $y$. Then by Lemma \ref{dBTprotocolMatchingilen}, $|w^*| \geq l_{i^*} - 3(k + \frac{t}{b_{^*}})$. The number of wrong matches in $w^*$ is at most $3(k + \frac{t}{b_{^*}})$ by Lemma \ref{dBTprotocolNWM}. So the total number of uncovered blocks is at most $6(k + \frac{t}{b_{^*}})$ in level $i^*$ and thus at most $12(k + \frac{t}{b_{^*}})$ in level $i^*+1$.
	
%	As a result, after the extra level, by Corollary \ref{dBTdocexcCorollary}, Bob can finally recover $x$.
	
\end{proof}

%\begin{corollary}
%\label{corollausedforcodesk}
%If the input strings for Alice have all three properties in Theorem~\ref{Bdistinctandinterval}, then there is a deterministic polynomial time one round document exchange protocol for $(k_1, t)$-block insertion/deletion and $k_2$ block transpositions with communication cost $O(k \log n \log \log n + t)$ bits and error probability $1 - 1/\poly(n)$, where $k = k_1+k_2$.
%\end{corollary}

%% file: BTcode.tex
\section{Binary codes for block edit errors}
\label{sec:BTcode}
\subsection{Encoding and decoding algorithm}
\label{EncDecAlgo}
Given the document exchange protocol for block edit operations, we can now construct codes capable of correcting $(k_1,t)$ block insertions/deletions, and $k_2$ block transpositions,  where $k_1+k_2=k$, and  $t\leq\alpha n$ for some constant $\alpha$. The encoding and decoding algorithms are as follows:

\begin{algo}{Encoding algorithm}
\label{EncodingAlgo}

Let $\ell_\mathsf{buf}=2\log n$ and $\mathsf{buf}=0^{\ell_\mathsf{buf}-1}\circ 1$.

\textbf{Input}: $\mathsf{msg}$ of length $n$.

\textbf{Ingredients}:

\begin{itemize}
\item A pseudorandom generator $\mathsf{PRG}: \{0,1\}^{O(\log n)} \rightarrow \{0,1\}^{n}$, from Theorem \ref{genCJLW18}, s.t. there exists at least one seed $r$ for which $msg\oplus \mathsf{PRG}( r)$ doesn't contain $\textsf{buf}$ as a substring and has B-distinctness.

\item An error correcting code $\mathcal{C}_1$ from Theorem \ref{asympGoodECCforInsdel}  which is capable of correcting $O(k\log n+t)$ edit errors, as well as $k_2$ block transpositions. Denote the encoding map of $\mathcal{C}_1$ as $Enc_1:\{0,1\}^{\mathsf{ml}_1 = O(k\log^2 n + t)} \rightarrow \{0,1\}^{\mathsf{cl}_1 =  O(k\log^2 n + t)}$ and the decoding map as $Dec_1:\{0,1\}^{\mathsf{cl}'_1} \rightarrow \{0,1\}^{\mathsf{ml}_1}$.
%\item A code $\mathcal{C}_2$ with encoding algorithm $Enc_2:\{0,1\}^{\log n} \rightarrow \{0,1\}^{2\log n}$ such that the codewords of $\mathcal{C}_2$ doesn't contain $\mathsf{buf}$ as a substring. Denote the decoding map of $\mathcal{C}_2$ as $Dec_2$.
\end{itemize}

\textbf{Operations}:

\begin{enumerate}

\item Find a seed $r$ of $\mathsf{PRG}$ s.t. $\mathsf{msg} \oplus \mathsf{PRG}(r)$ does not contain $\mathsf{buf}$ as a substring and satisfies B-distinctness. Let $\mathsf{msg}_P=\mathsf{msg}\oplus \mathsf{PRG}(r)$.
\item Compute the sketch $\mathsf{sk}_m$ for $\mathsf{msg}_P$ for  $\Omega(k)$ block insertions/deletions and $\Omega(k)$ block transpositions, where the number of bits inserted and deleted is $\Omega(k\log n+t)$ in total. \label{enc_op2}
\item Let $\mathsf{sk}=\mathsf{sk}_m\circ r$, and encode $\mathsf{sk}$ with $\mathcal{C}_1$. Let the codeword be $c_1=Enc_1(\mathsf{sk})$.
\item Divide $c_1$ into blocks of length $\log n$. Denote these blocks as $c_1^{(1)},c_1^{(2)},\dots, c_1^{(M)}$ where $M$ is the number of blocks.
\item Insert $\mathsf{buf}$ to the beginning of each block $c_1^{(i)},1\leq i\leq M$.
\item Let $c = (\mathsf{msg}\oplus \mathsf{PRG}(r))\circ \mathsf{buf}\circ c_1^{(1)} \circ \mathsf{buf} \circ c_1^{(2)}\dots \circ \mathsf{buf} \circ c_1^{(M)}$.
\end{enumerate}

\textbf{Output}: $c$.
\end{algo}

The construction of $\mathsf{PRG}$ is left to subsection \ref{EncDecAnalysis}. We call the concatenation $\mathsf{buf}\circ c_1^{(1)} \circ \mathsf{buf} \circ c_1^{(2)}\dots \circ \mathsf{buf} \circ c_1^{(M)}$ as the sketch part and $\mathsf{msg}_P=\mathsf{msg}\oplus \mathsf{PRG}(r)$ as the message part. Now we give the corresponding decoding algorithm.

\begin{algo}{Decoding algorithm}
\label{DecodingAlgo}

\textbf{Input}: the received codeword $c'$.

\textbf{Operations}:

\begin{enumerate}
\item Find out all substrings $\mathsf{buf}$ in $c'$. Number these buffers as $\mathsf{buf}_1,\dots, \mathsf{buf}_{M'}$.
\item Pick the $ \log n$ bits after $\mathsf{buf}_j$ as block ${c_1'}^{(j)}, 1\leq j\leq M'$. Then remove all the buffers $\mathsf{buf}_j$ and  ${c'_1}^{(j)}$, $1\leq j\leq M'$ from $c'$. The rest of $c'$ is regarded as the message part $\mathsf{msg}_P^{'}$.
\item Let $c_1'={c'_1}^{(1)}\circ {c'_1}^{(2)}\circ\dots\circ {c'_1}^{(M')}$. Decode $c_1'$ with the decoding algorithm $Dec_1$ for $\mathcal{C}_1$ and get $\mathsf{sk}=Dec_1(c_1')$.
\item Get $\mathsf{sk}_m$ and $r$ from $\mathsf{sk}$.
\item Use $\mathsf{sk}_m$ and $\mathsf{msg}_P^{'}$ to recover $\mathsf{msg}_P$. \label{dec_op5}
\item Compute $\mathsf{msg} = \mathsf{msg}_P \oplus \mathsf{PRG}(r)$.
\end{enumerate}

\textbf{Output}: $\mathsf{msg}$.
\end{algo}
\subsection{Analysis}
\label{EncDecAnalysis}
In this subsection we'll give the construction of $\mathsf{PRG}$ and prove the correctness of the algorithms.

\subsubsection{Building blocks: $\mathbf{PRG}$}

We recall the pseudorandom generator in Theorem 5.1 in \cite{CJLW18}.
\begin{theorem}[Theorem~5.1 in \cite{CJLW18}]\label{genCJLW18}
For every $  n \in \mathbb{N}$,  there exists an explicit PRG $g:\{0,1\}^{\ell = O(\log n)} \rightarrow \{0,1\}^{n} $ s.t. for every $ x\in \{0,1\}^n$,  with probability $1-1/\poly(n)$,  $g(  U_{\ell}) + x $ satisfies B-distinctness.
\end{theorem}

\begin{theorem}
	For every $n \in \mathbb{N}, x \in \{0, 1\}^n$,there exists an explicit PRG $g:\{0,1\}^{\ell = O(\log n)} \rightarrow \{0,1\}^{n} $ s.t. for every  $ x\in \{0,1\}^n$, with probability $1 - 1 / \poly(n)$, the following two conditions hold simultaneously.
	\begin{itemize}
		\item $\mathsf{buf}$ is not a substring of $\mathsf{PRG}( U_{\ell}) \oplus x$.
		\item $\mathsf{PRG}(U_{\ell}) \oplus x$ satisfies B-distinctness.
	\end{itemize}
\end{theorem}

\begin{proof}
	Let $\kappa = \ell_{\mathsf{buf}}$ be the length of $\mathsf{buf}$, $\eps = 1/n^2$.
	From Theorem~\ref{almostkwiseg}, there exists an
	explicit $\eps$-almost $\kappa$-wise independence generator $g': \{0, 1\}^d \rightarrow \{0, 1\}^n$, where $d = O(\log \frac{\kappa \log n }{\eps}) = O(\log n)$.
	Then, for any $x \in \{0, 1\}^n$,
	\begin{align*}
	\Pr_{r' \gets \{0, 1\}^d}[\mathsf{buf} \text{ is a substring of } g'(r') \oplus x] &
	\le \sum_{i \in [n - \ell_{\mathsf{buf}} + 1]}\Pr[\mathsf{buf} = (g'(r') \oplus x)[i, i + \ell_{\mathsf{buf}})] \\
	&\le n \left(1/2^{\ell_{\mathsf{buf}}} + 1/n^2\right) = 1/\poly(n)
	\end{align*}
	
	Let $g$ be the generator in Theorem~\ref{genCJLW18} with seed length $\ell'$.
	Let $\ell = \max(\ell', d)$,
	and construct $\mathsf{PRG}(  r) = g( r_1) \oplus g'( r_2)$ where $r_1, r_2$ are disjoint substrings of $r$ of length $l'$ and $l$.
	Then by the union bound, the probability that at least one of the conditions fails is upper bounded by $1 / \poly(n)$.
\end{proof}

%\subsubsection{Building blocks: $\mathbf{\mathcal{C}_1}$}
%
%Fix a constant $\eps$ such that $\eps > \frac{12k\log n+t}{3|\mathsf{sk}|}$. We use the code in \ref{asympGoodECCforInsdel} that can correct $\eps$  fraction of insertions and deletions, as well as $k$ block transpositions. The rate of $\mathcal{C}_1$ is a constant.

%\subsubsection{Building blocks: $\mathbf{\mathcal{C}_2}$}
%\begin{theorem}
%There exits a coding scheme $(Enc_2, Dec_2)$ such that
%$Enc_2: \{0, 1\}^{\log n} \rightarrow \{0, 1\}^{2 \log n}$
%satisfies that for every $c \in \{0, 1\}^{\log n}$, \mathsf{buf} is not a substring of $Enc_2(c)$. The $Enc_2, Dec_2$ both runs in $\poly(n)$ time.
%\end{theorem}
%
%\begin{proof}
%	On input $c \in \{0, 1\}^{\log n}$, the encoding algorithm firstly parses $c$ as an integer in $[n]$, then lists all strings in $\{0, 1\}^{2 \log n}$ that don't contain the \mathsf{buf} as substring in lexicographic order, and choose $c$-th string. The decoding algorithm reverses the encoding procedure accordingly. Since there are $n^2 - 1 > n$ strings in $\{0, 1\}^{2 \log n}$ that isn't equal to \mathsf{buf}, such encoding algorithm is feasible.
%\end{proof}

\subsubsection{Correctness of the construction}
We show that a code $\mathcal{C}$ with encoding algorithm \ref{EncodingAlgo} and decoding algorithm \ref{DecodingAlgo} can correct $(k_1,t)$-block insertions/deletions and $k_2$ block transpositions.

First, we prove the sketch $\mathsf{sk}$ can be correctly recovered.
\begin{lemma}
\label{skDecode}
In the 4th step of decoding algorithm \ref{DecodingAlgo}, the sketch $\mathsf{sk}$ is correctly recovered.
\end{lemma}

\begin{proof}
We show that $c_1'$ can be obtained by applying at most $12k\log n+t$ edit errors and $k$ block transpositions over $c_1$.

Note that after inserting buffers to the blocks of $c_1$, the total number of appearance of the buffer in the sketch part is equal to the number of buffers inserted, because the buffer length is longer than the block length of $c_1$.
Also note that concatenating the message part and sketch part will not insert any buffers because by the choice of $r$, $\mathsf{msg} \oplus \mathsf{PRG}(r)$ does not contain $\mathsf{buf}$. As a result, if there are no errors, by the decoding algorithm we can get the correct $c_1$ and thus get the correct $\mathsf{sk}$.

Next we consider the effects of block insertions/deletions and transpositions for the sketch part. Specifically, we consider how the sketch part changes after each of these operations.

\begin{itemize}
\item block insertion: Consider one block insertion of $t_0$ bits. We claim that after this operation,  at most $ \lceil t_0/(3\log n) \rceil  $ new blocks can be introduced to the sketch part, because to insert one new block to the sketch, we only need to insert a new buffer and attach the new block to it.
We also note that this operation may delete one block by damaging a buffer, or replace one block by damaging the block  right after the buffer.

So $k_1$ block insertions of $t$ bits inserted can insert at most $k_1 + t/(3\log n)$ new blocks. It can also delete at most $k_1$ blocks, and replace at most $k_1$ blocks.

\item block deletion:
we first consider a block deletion of $t_0$ bits. After this operation, at most $ \lceil t_0/(3\log n) \rceil $   blocks of the sketch part can be deleted, since there are at most  $ \lceil t_0/(3\log n) \rceil $   blocks in the deleted substring. The operation may also create one extra block, since the remaining bits may combine together to be a buffer. It may also replace an existing block, since the remaining bits may combine together to be a new block after an original buffer.

So $k_1$ block deletions of $t$ bits deleted can delete at most $k_1 + t/(3\log n)$ blocks. It can  insert at most $k_1$ blocks. It can also replace $k_1$ blocks.

\item block transposition:
After one block transposition $(i,j,l)$,  at most $3$ new blocks can be introduced to the sketch part, since a new block may be created at the original position $i$, and two new blocks may appear when inserting the block to the destination $j$. Also it may delete at most $3$ blocks, since two buffers may be damaged when removing the transferred block, and one buffer can be damaged when inserting the transferred block. By a similar argument this operation can replace at most $3$ blocks.
Also, a block transposition can cause one block transposition for the sketch part.

As a result, $k_2$ block transpositions can insert or delete at most $O(k_2)$ blocks  and cause $O(k_2)$ block transpositions.

\end{itemize}

In summary, there are at most $ O(k + t/\log n) $ block insertions/deletions and $k_2$ block transpositions on $c_1$. Note that $ O(k  + t/\log n) $ block insertions/deletions, each of length $O(\log n)$ bits can be regarded as  $ O(k \log n + t) $ edit errors. Since our code $\mathcal{C}_1$ can correct  $ O(k \log n + t) $  edit errors and $k_2$ block transpositions, we can decode $sk$ correctly.

\end{proof}

Next, we show that the message output by the decoding algorithm is correct.
\begin{lemma}
\label{msgDecode}
At the end of algorithm \ref{DecodingAlgo}, the original message is correctly decoded.
\end{lemma}

\begin{proof}
According to Lemma \ref{skDecode}, we have correctly recovered $\mathsf{sk}$. Thus we get $\mathsf{sk}_m$ and $r$ correctly.

Note that if there are no errors, then by deleting the buffers and the blocks of $c_1$ appended to these buffers, the remaining string is exactly the original message part, since the original message part does not contain $\mathsf{buf}$ as substrings.

Now we consider the effects of block insertions/deletions and transpositions for the message part. Specifically, we consider how the message part changes after each of these operations.
%\begin{itemize}
%\item block insertion:
%First consider one block insertion of $t_0$ bits.
%It can insert at most $ \lceil t_0/\log n \rceil $ new blocks to the message part if it does not damaging any original buffers. If it damages buffers, it can cause one block insertion of $ t_0 + O(\log n)$ bits in the message part. It can also cause at most one block deletion of $O(\log n)$ bits since the rightmost buffer it inserts may cause our algorithm to delete the $O(\log n)$ bits following that buffer.
%
%\item block deletion:
%Consider a block deletion of $t_0$ bits. It can delete at most $ \lceil t_0/\log n \rceil $   blocks of the message part. If it damages buffers, it can cause at most one block insertion of $O(\log n)$ bits, since the rightmost deleted buffer may cause our algorithm to regard the $O(\log n)$ bits following that buffer as part of the message part.
%
%Thus $(k_1, t)$-block insertions/deletions  can cause inserting/deleting at most $O(k_1)  $ blocks of $O(t+ k_1 \log n)$ bits.
%
%Next we consider one block transposition. It may cause at most one block transposition of the message part. Also it may create at most 3 new buffers and thus delete $3$ blocks,  each of $O(\log n)$ bits, in the message part. Moreover, it may delete three buffers and thus insert $3$ blocks, each of $O(\log n)$ bits, to the message part.
%
%So $k_2$ block transpositions can cause $O(k_2)$ block insertions/deletions of $O(k_2\log n)$ bits in total and $k_2$ block transpositions.
%
%\end{itemize}

\begin{itemize}
\item block insertion:
First consider one block insertion of $t_0$ bits.
It can insert at most $ t_0   $ symbols to the message part if it does not damaging any original buffers. If it damages buffers, it may insert  $ O(\log n)$ more bits to the message part. It can also cause at most one block deletion of $O(\log n)$ bits since the rightmost buffer it inserts may cause our algorithm to delete the $O(\log n)$ bits following that buffer.

\item block deletion:
Consider a block deletion of $t_0$ bits. It can delete at most $ t_0$   blocks of the message part. If it damages buffers, it can cause at most one block insertion of $O(\log n)$ bits, since the rightmost deleted buffer may cause our algorithm to regard the $O(\log n)$ bits following that buffer as part of the message part.

\item block transposition:
now we consider one block transposition. It may cause at most one block transposition of the message part. Also it may create at most 3 new buffers and thus delete $3 \log n$ bits of the message part. Moreover, it may delete three buffers and thus insert $3 \log n$ bits to the message part.

\end{itemize}

Thus $(k_1, t)$-block insertions/deletions  can cause inserting/deleting at most $ O(k_1)  $ blocks of $O(t+ k_1 \log n)$ bits.
Also $k_2$ block transpositions can cause $O(k_2) $ block insertions/deletions of $O(k_2\log n)$ bits in total and $k_2$ block transpositions.

In summary, there are at most $O(k)$ block insertions/deletions of $O(k\log n + t)$ bits in total and $k_2$ block transpositions. Since our sketch $\mathsf{sk}$ can be used to correct  $(O(k),O(k\log n+t))$ block insertions/deletions and $k$ block transpositions, we can get $\mathsf{msg}_P$ correctly. As a result we can compute $\mathsf{msg} = \mathsf{msg}_P \oplus \mathsf{PRG}(r)$ correctly.

\end{proof}

\begin{theorem}
%The code $\mathcal{C}$ with encoding algorithm \ref{EncodingAlgo} and decoding algorithm \ref{DecodingAlgo} can correct  $(k_1,t)$  block insertions/deletions and $k_2$ block transpositions where $k_1, k_2, t\leq \alpha n$ for some constant $\alpha$. The redundancy length of $\mathcal{C}$ is $k\log n\log\log n+t$ where $n$ is the message length.

For every $n, k_1, k_2, t \in \mathbb{N}$ with $k = k_1+ k_2 < \alpha n/\log n, t \leq \beta n$, for some constant $\alpha, \beta$, there exists an explicit binary error correcting code for $(k_1, t)$-block insertions/deletions and $k_2$
block transpositions, having message length $n$, codeword length $ n +   O(k\log n \log\log\log n+t)$.
\end{theorem}

\begin{proof}
%We show that, by using the document exchange protocol of Theorem \ref{corollausedforcodesk} in  Operation~\ref{enc_op2} of the Algorithm~\ref{EncodingAlgo} and Operation~, we can have such a code.

We construct the encoding   as Algorithm~\ref{EncodingAlgo} where the sketch in Stage~\ref{enc_op2} is computed by using Alice's algorithm (encoding) of the protocol of Theorem \ref{randmainthm}. The decoding is as Algorithm~\ref{DecodingAlgo}, where its stage~\ref{dec_op5} is computed by using Bob's algorithm of the protocol of Theorem \ref{randmainthm}.

The correctness of the decoding algorithm in \ref{EncodingAlgo} is shown by Lemma \ref{msgDecode}.

The sketch length $|\mathsf{sk}|$ is $O(k\log n\log\log \log n+t)$ by Theorem \ref{randmainthm}. The length of $c_1 $ is $O(|\mathsf{sk}|) = O(k\log n \log\log\log n+t)$ by Theorem \ref{asympGoodECCforInsdel}.
As there are $|c_1|/\log n$ number of length $l_{\mathsf{buf}}= O(\log n)$ buffers, each followed by a  length $\log n$ block of $c_1$, the total length of the sketch part is $O(|c_1|) = O(k\log n\log\log\log n+t)$.
\end{proof}

We can also directly using our document protocol to get an ECC.
\begin{theorem}
For every $n, k_1, k_2, t \in \mathbb{N}$ with $k = k_1+ k_2 < \alpha n/\log n, t \leq \beta n$, for some constant $\alpha, \beta$, there exists an explicit binary error correcting code for $(k_1, t)$-block insertions/deletions and $k_2$
block transpositions, having message length $n$, codeword length $ n +   O(( k  \log n +t) \log^2 \frac{n}{k\log n + t}  )$.
\end{theorem}

\begin{proof}

We construct the encoding   as Algorithm~\ref{EncodingAlgo} where the sketch in Stage~\ref{enc_op2} is computed by using Alice's algorithm (encoding) of the protocol of Theorem \ref{dBTdocexc}. The decoding is as Algorithm~\ref{DecodingAlgo}, where its stage~\ref{dec_op5} is computed by using Bob's algorithm of the protocol of Theorem \ref{dBTdocexc}.
	
	The correctness of the construction is similar to Lemma~\ref{skDecode}, \ref{msgDecode}, the $(k_1, t)$-block insertions/deletions and $k_2$ block transpositions causes $(k, O(k \log n + t))$-block insertions/deletions and transpositions on the message and sketch part. Hence, according to  Theorem \ref{dBTdocexc}, a sketch of size $O((k \log n+t) \log^2\frac{n}{k \log n + t})$ for the document exchange protocol is enough to correct the errors.
	
	By Algorithm \ref{EncodingAlgo} and Theorem~\ref{asympGoodECCforInsdel}, the size of $c_1$ is  $O((k \log n+t) \log^2\frac{n}{k \log n + t})$ .
	The total length of the buffer inserted is $O(\log n) \cdot |c_1|/\log n = O(|c_1|)$. Hence, the total length of the redundancy is  $O((k \log n+t) \log^2\frac{n}{k \log n + t})$ .
\end{proof}

%% file: appendix.tex
\appendix
\begin{center}
\bfseries \huge Appendices
\end{center}
\section{}
\label{appendix}

\begin{theorem}\label{docexlowerbnd}
	Suppose there is a deterministic document exchange protocol for strings of length $n$,
	and can resist $k$ block insertions/deletions and block transposition errors,
	where the total number of bits inserted or deleted is bounded by $t$, and $t < n/2$,
	then the sketch size is at least $\Omega(k \log n + t)$.
\end{theorem}

\begin{proof}
	Suppose Alice has string $x$ and Bob has string $y$, and
	Alice sends a sketch $\mathsf{sk}(x)$ to allow Bob recovering her string $x$.
	For a fixed string $y$,
	each different strings $x_1, x_2$ satisfy $\mathsf{sk}(x_1) \neq \mathsf{sk}(x_2)$, otherwise the correctness of the document exchange protocol will be violated.
	Now suppose $y$ is a fixed string of length $n$ satisfying
	$B$-distinct property, where $B = O(\log n)$, we give a lower bound on the number of possible strings of $x$.
	
	Consider the following adversarial tempering of the string $x$:
	delete the last $t/2$ bits as a block, then insert arbitrary $t/2$ bits at the end as a block. Next, divide the $(n-t/2)$-prefix evenly to small blocks of length $B$.
	Arbitrary choose $k-2$ different small blocks and transpose them to the begining of the string in an arbitrary order.
	Then any differences in the $t/2$ bits inserted, the choice of the blocks or the ordering will result to different strings.
	Hence, the number of strings $x$ is lower bounded by
	\begin{align*}
	2^{t/2} \begin{pmatrix}
	\frac{n - t/2}{B} \\
	k-2
	\end{pmatrix} k! \ge 2^{t/2} \left(\frac{3n}{4(k-2)B}\right)^{k-2} \left(\frac{k-2}{e}\right)^{k-2} = 2^{t/2} \left(\frac{3n}{4eB}\right)^{k-2}
	\end{align*}
	
	Taking the $\log$, we obtain $|\mathsf{sk}| \ge \Omega(k \log n + t)$.
\end{proof}

\begin{theorem}\label{redlowerbnd}
	Let $n', n$ be two integers,
	if $\mathcal{C} \subseteq \{0, 1\}^{n'}, |\mathcal{C}| = 2^n$ is an Error Correcting Code for $k$ block insertions/deletions and block transpositions,
	where the total number of bits inserted or deleted is bounded by $t$, and $t < n/100$,
	then the redundancy size $n' - n \ge \Omega(k \log n + t)$.
\end{theorem}

\begin{proof}
	Denote $n'' = n' - t/2$.
	It suffices to consider the case $n'' < 2n$.
	We evenly divide the interval $[1, n'')$ into smaller intervals of length $10 \log n''$, and denote these intervals as
	$I_1, I_2, \dots, I_{n''/10 \log n''}$.
		
	Let $\mathcal{C}'$ be a subset of $\mathcal{C}$ containing all the codewords $c$ such that the number of distinct strings in $\{c_{I_1}, c_{I_2}, \dots, c_{I_{n''/10 \log n''}} \}$ is at least $n'' / 1000 \log n''$. We will show that $\mathcal{C}'$ contains a large fraction of the codewords.

	For simplicity, we denote $a = n'' / 1000 \log n''$.
	Now we bound the size of the set $\mathcal{C} \setminus \mathcal{C}'$.
	Note that any codewords $c \in \mathcal{C} \setminus \mathcal{C}'$
	satisfies that the number of distinct strings in $\{c_{I_{1}}, c_{I_{2}}, \dots, c_{I_{n''/10 \log n''}}\}$ is smaller than $a$.
	Hence we have
	\begin{align*}
	|\mathcal{C} \setminus \mathcal{C}'|
	&\le a^{n'' / 10 \log n''} (2^{10 \log n''})^a 2^{t/2}
	= 2^{n'' \log a / 10 \log n'' + 10 a \log n'' + t/2} \\
	&\le 2^{n'' / 10 + n'' / 100 + t/2} \le 2^{3 n / 5}
	\end{align*}
	
	Now we obtain the lower bound of $|\mathcal{C}'|$. When $n \ge 2$,
	\begin{align*}
	|\mathcal{C}'| = |\mathcal{C}| - |\mathcal{C} \setminus \mathcal{C}'| \ge 2^n - 2^{3n / 5} \ge 2^{n} / 2
	\end{align*}
	
	For any codeword $c \in \mathcal{C}'$, define the ball $\mathcal{B}_c(k, t)$ to be the set
	containing all strings obtained by applying $k$ block insertions/deletions and block transpositions to $c$, where the total number of bits inserted or deleted is bounded by $t$.
	
	Consider the following adversarial tempering of the codeword $c$: delete the last $t/2$ bits of $c$ as a block deletion, then insert arbitrary $t/2$ bits at the ending of the tempered string as a block insertion. Next, arbitrary choose $k-2$ distinct strings from $c_{I_1}, c_{I_2}, \dots, c_{I_{n''/10 \log n''}}$,
	and transport them to the begining of the string in an arbitrary order.
	Then, any differences in the $t$ bits inserted, the choice of the $(k-2)$ substrings or the order of transpositions will result in different strings in $\mathcal{B}_c(k, t)$.
	Hence,
	
	\begin{align*}
		|\mathcal{B}_c(k, t)| \ge
		2^{t/2} \begin{pmatrix}
			a \\
			k-2
		\end{pmatrix} (k-2)!
		\ge
		2^{t/2} \left(\frac{n''/1000 \log n''}{k-2}\right)^{k-2} \left(\frac{k-2}{e}\right)^{k-2} =
		2^{t/2}\left(\frac{n''}{1000e \log n''}\right)^{k-2}
	\end{align*}
	
	As $\mathcal{C}$ is a code, the ball $\mathcal{B}_c(k, t)$ should be disjoint,
	so we have
	
	\begin{align*}
	2^{n'} \ge \sum_{c \in \mathcal{C}} |\mathcal{B}_c(k, t)|
	\ge |\mathcal{C'}| 2^{t/2} \left( \frac{n''}{1000e \log n''} \right)^{k-2} \ge 2^{n + t/2 - 1} \left( \frac{n''}{1000e \log n''} \right)^{k-2}
	\end{align*}
	
	Taking a $\log$ on both sides of the equation, we obtain $n' \ge n + \Omega(k \log n'' + t) \ge n + \Omega(k \log n + t)$.
\end{proof}

\begin{theorem}
	There exists a deterministic document exchange protocol running in exponential time in $n$ with sketch size $O(k \log n + t)$.
	Moreover, we can construct an Error Correting Code with redundancy size $O(k \log n + t)$ from the document exchange protocol.
	Hence the lower bounds in Theorem~\ref{docexlowerbnd} and Theorem~\ref{redlowerbnd} are tight.
\end{theorem}

\begin{proof}
	We build a graph. Each string with length smaller than $n + t$ corresponds to a vertex in the graph.
	For every two different strings $x$ and $y$, if one can transform $x$ to $y$ using $k$ block insertions/deletions and transpositions, and the total number of inserted and deleted bits is bounded by $t$, then add an edge between $x$ and $y$.
	Now the degree of the graph is at most $(2n)^{O(k)} 2^t = 2^{O(k \log n + t)}$, hence we can use $2^{O(k \log n + t)}$ colors to color the graph.
	
	We construct the document exchange protocol as follows.
	Given the input string, Alice sends the color of the string as the sketch, so the sketch has size $O(k \log n + t)$ bits.
	Then Bob looks at the strings connected to his string, and
	find the string whose color matches the sketch.
	
	In fact, the construction of the Error Correcting Code in Section~\ref{sec:BTcode} can be applied to any document exchange protocol, so we obtain an Error Correcting Code of redundancy $O(k \log n + t)$.
\end{proof}